\documentclass[journal,draftcls,onecolumn,12pt,twoside]{IEEEtran}

\usepackage[ruled,lined,linesnumbered]{algorithm2e}

\ifCLASSINFOpdf

\else

\fi

\usepackage{amsfonts,epsfig,cite,array,multirow,graphicx,amsmath,ltablex,tabularx,setspace,amssymb,multirow,stackengine}
\usepackage{schemabloc,amsthm}
\usepackage[subrefformat=parens,labelformat=parens]{subfig}
\usepackage[numbers,sort&compress,comma,numbers]{natbib}
\usetikzlibrary{circuits, arrows}
\usepackage{verbatim}
\usepackage{graphicx}
\usepackage{epstopdf}
\usepackage{mathtools}
\usepackage{xcolor}
\usepackage{marvosym}
\usepackage{calc}
\usepackage{aurical}
\usepackage{hhline}
\usepackage{floatrow}
\newfloatcommand{capbtabbox}{table}[][\FBwidth]
\usepackage{caption}
\usepackage{blindtext}
\usepackage{lipsum}
\usepackage{dblfloatfix}
\usepackage{color}
\usepackage{epsfig}
\usepackage{pst-grad} 
 \usepackage{pst-plot} 
\usepackage{auto-pst-pdf}
\usepackage[shortlabels]{enumitem}
\theoremstyle{plain}

\newtheorem{lemma}{Lemma}

\newtheorem{rem}{Remark}
\newtheorem{exm}{Example}
\usepackage[font=footnotesize]{caption}

\usepackage{collectbox}
\DeclareFontFamily{U}{mathx}{\hyphenchar\font45}
\DeclareFontShape{U}{mathx}{m}{n}{
	<5> <6> <7> <8> <9> <10>
	<10.95> <12> <14.4> <17.28> <20.74> <24.88>
	mathx10
}{}
\DeclareSymbolFont{mathx}{U}{mathx}{m}{n}
\DeclareMathSymbol{\bigtimes}{1}{mathx}{"91}
\makeatletter
	\definecolor{darkblue}{rgb}{0.0, 0.0, 0.55}
\theoremstyle{remark}
 
\usepackage{hyperref}
\usepackage{placeins}
\definecolor{darkorange}{rgb}{1.0, 0.55, 0.0}
\hypersetup{colorlinks=true,  urlcolor=black, linkcolor=black, citecolor=black}
\usepackage[all]{hypcap}
\stackMath
\begin{document}
\SetKwRepeat{Do}{do}{while}
\newcolumntype{L}[1]{>{\raggedright\let\newline\\\arraybackslash\hspace{0pt}}m{#1}}
\newcolumntype{C}[1]{>{\centering\let\newline\\\arraybackslash\hspace{0pt}}m{#1}}
\newcolumntype{R}[1]{>{\raggedleft\let\newline\\\arraybackslash\hspace{0pt}}m{#1}}
\newcommand{\oeq}{\mathrel{\text{\sqbox{$=$}}}}
\setlength{\textfloatsep}{0.1cm}
\setlength{\floatsep}{0.1cm}
\title{
Convolutional Sparse Coding based Channel Estimation for OTFS-SCMA in Uplink
\vspace{-.1em} }

\author{
    \IEEEauthorblockN{\small{ Anna Thomas$^{1}$, Kuntal Deka$^{1}$, P. Raviteja$^{2}$, and Sanjeev Sharma$^{3}$ } }\\
	$^1$Indian Institute of Technology Goa, India,  $^2$Qualcomm, San Diego, California, United States,\\   $^3$Indian Institute of Technology (BHU) Varanasi, India  \vspace{-0.4in}}

\maketitle

\begin{abstract}
Orthogonal time frequency space (OTFS) has emerged as the most sought-after modulation technique in a high mobility scenario. Sparse code multiple access (SCMA) is an attractive code-domain non-orthogonal multiple access (NOMA) technique. Recently a code-domain NOMA approach for OTFS, named OTFS-SCMA, is proposed. OTFS-SCMA is a promising framework that meets the demands of high mobility and massive connectivity. This paper presents a channel estimation technique based on the convolutional sparse coding (CSC) approach for OTFS-SCMA in the uplink. The channel estimation task is formulated as a CSC problem following a careful rearrangement of the OTFS input-output relation. We use an embedded pilot-aided sparse-pilot structure that enjoys the features of both OTFS and SCMA. The existing channel estimation techniques for OTFS in multi-user scenarios for uplink demand extremely high overhead for pilot and guard symbols, proportional to the number of users. The proposed method maintains a minimal overhead equivalent to a single user without compromising on the estimation
error. The results show that the proposed channel estimation algorithm is very efficient in bit error rate (BER), normalized mean square error (NMSE), and spectral efficiency (SE).
\end{abstract}


%
\section{Introduction}

\subsection{Motivation}
In current communication standards for 4G  and 5G,  orthogonal frequency division multiplexing (OFDM) has been undisputedly recommended for modulation. OFDM was particularly designed to eliminate the inter-symbol-interference (ISI) caused by the time dispersion of the channel. The success of OFDM depends on the orthogonality of the sub-carriers.  If the channel introduces frequency dispersion too in the form of Doppler shifts, the  orthogonality of the sub-carriers is destroyed.  This issue creates inter-carrier-interference (ICI),  which hinders the application of OFDM in high Doppler scenario.
The coming generation of wireless networks implicates high Doppler due to increased mobility in environments like vehicle-to-everything (V2X) and high carrier frequency such as in mmWave communication. Orthogonal time frequency space (OTFS) has turned up as a solution to boost
up the potential of such communication scenario~\cite{Hadani}. The remarkable performance of OTFS is attributed to the two-dimensional (2D) modulation technique in which both the data and the channel are represented in delay-Doppler (DD) domain, by exploiting the quasi-periodicity property in this domain. OTFS modulation benefits from  the existence of 2D localised pulses in DD domain.  These pulses subsequently occupy the entire time-frequency (TF) grid, thereby achieving full diversity. All the symbols in an OTFS DD frame experience nearly the same channel. Moreover, fewer parameters are required to describe the channel in the DD domain facilitating a sparse representation, which makes the channel estimation an easy task.     


The superior performance of OTFS is evident in multi-user scenarios also.   There can be two approaches: (1) orthogonal multiple access (OMA) and (2)~non-orthogonal multiple access (NOMA). In NOMA, various users share resources, unlike in OMA. The  spectral efficiency of NOMA is significantly better than that of OMA.  OTFS-OMA with different interleaving patterns was proposed  in \cite{surabhi_MA,Khammammetti2019a}.  OTFS with  power-domain NOMA  was explored in \cite{Ding2019,pd_noma}. Recently, we proposed  a code-domain NOMA approach for OTFS based on sparse code multiple access (SCMA), named   OTFS-SCMA \cite{otfs_scma}. The simulation results and diversity analysis showed the superiority of OTFS-SCMA  over  other multi-user techniques of OTFS. Any channel estimation technique for OTFS can be easily extended to OTFS-SCMA in  downlink, unlike in  uplink. Thus, it is highly relevant to devise a channel estimation technique for OTFS-SCMA in  uplink. 

\subsection{Related Prior Works}\label{par:related_works}
The   detection of an OTFS-based  system  necessitates  an accurate channel estimation. 
OTFS-SCMA integrates the techniques of modulation of OTFS and the multiple access scheme of  SCMA. Hence, to devise the channel estimation algorithm of OTFS-SCMA, the conventional methods used for both OTFS and SCMA are to be analyzed. The commonly-adopted   approach  of channel estimation for OTFS 
is the pilot-aided method in \cite{hadani_ce}.  With the help of the pilot symbols placed in the DD grid, the channel coefficients are estimated.  Following  this approach, the authors in \cite{ch_est1} presented a systematic method for  channel estimation using an embedded single QAM pilot symbol with a guard band. By applying thresholding in the observation region, paths are identified,  and the channel coefficients are estimated by element-by-element division.     The same method was also extended to   the  MIMO and multiuser  cases  \cite{ch_est1}. In \cite{sbl_ce}, the authors adopted a sparse signal recovery approach  using the sparse Bayesian learning (SBL) algorithm. In this method,  the embedded pilot structure does not contain a guard band; instead, it has multiple  QAM pilot symbols.
Reference \cite{ch_est3} also follows a similar pilot-aided technique for channel estimation of OTFS-MIMO.   Channel estimation techniques for massive MIMO-OTFS are developed in \cite{CE_heath,ce_mm}, where  a 3D channel model of delay-Doppler-angle is considered.  To estimate the channel parameters,   3D-structured Orthogonal Matching Pursuit (OMP) and SBL algorithms are used in \cite{CE_heath} and \cite{ce_mm}, respectively.   The channel estimation for OTFS-OMA is carried out in \cite{ce_oma} using a sparse signal recovery approach based on OMP and subspace pursuit (SP). An embedded QAM pilot symbol with a guard band is used for OTFS-OMA channel estimation in \cite{surabhi_MA}. However,  these methods either need to dedicate an entire OTFS frame for the training of the  pilot symbols or  require significantly large pilot and guard band overhead which is proportional to the number of transmitting antennas or users. 
In the context of  the SCMA channel estimation, the reference \cite{scma_ce_rev}  considered both pilot-aided and data-aided approaches and analyzed  the  trade-off between the spectral efficiency and pilot overheads.  Active user detection  followed by channel estimation for SCMA was investigated in \cite{scma_ce1, scma_ce2}. The  method   in  \cite{scma_sparse_ce} considered  sparse pilot vectors with the  non-zero pilot symbols  placed according to the pattern of the corresponding codebook of a user.  
\subsection{Contributions}\label{par:contributions}
This paper proposes a channel estimation method for OTFS-SCMA based on convolutional sparse coding (CSC).  
This method is found to provide  impressive  bit-error-rate (BER) performance with minimal overhead and complexity.   The   special features and the main  contributions of the proposed method are summarized below:
\begin{itemize} 
	\item This work is the first attempt to perform channel estimation for OTFS in NOMA environment. Although OTFS was studied in the context of power-domain NOMA in \cite{OTFS_NOMA_Ge,Ding2019,pd_noma}, no channel estimation techniques were proposed.
	\item The pilot structure of the proposed method is embedded, meaning that the pilot symbols and the data symbols are transmitted in the  same frame.  Also, the proposed pilot vectors of the users are  sparse and  non-orthogonal as  they  follow the same sparsity pattern as that of data vectors.  The sparsity reduces the interference amongst pilot vectors of multiple users. 
	\item The non-orthogonal pilot structure helps to maintain a very low pilot and guard band overhead. In \cite{ce_oma}, the guard band is absent, but it requires a dedicated frame for pilot symbols. While in \cite{ch_est1} and \cite{surabhi_MA}, the overhead needed for channel estimation is proportional to the number of users, in our proposed method the guard band requirement is the same as that of a single user. 
	\item The CE for OTFS-SCMA in uplink is formulated as a  CSC problem.    This formulation   significantly reduces the dimensionality of the problem  facilitating the use of low-complexity SP-based recovery algorithms. 
	\item For solving the CSC-based channel estimation problem, SP is considered. SP demands the knowledge of sparsity. By exploring the properties of the standard propagation channel models and the structures of  the pilot and the guard band, the unknown sparsity of paths is converted to a known sparsity of the number of users.
	\item The initialization of the estimates needed for SP is done sequentially for each user in a greedy manner, resulting in fast convergence of the algorithm. Furthermore, the simulation-based observations indicate that the modified initialization reduces the  probability of false alarm and miss-detection of paths. Hence the proposed method performs close to that of the perfect channel state information (CSI). 
		\item Mutual coherence is an essential property of any compressive-sensing approach. It is the maximum correlation of any two dictionary elements or columns. The lesser the value of mutual
		coherence, the more is the chance of successful recovery. The pilot vectors are designed such that the mutual coherence of the dictionary is minimized. For this optimization of the pilot
		vectors, differential evolution \cite{DE_price_20} is considered.
\end{itemize}
\subsection{Outline} Section~\ref{sec::prelim} describes the  preliminaries of OTFS and SCMA.  This section also presents various features of  OTFS-SCMA and the concept of CSC.  The proposed channel estimation technique and its analysis are presented in  Section~\ref{sec:proposed}, highlighting the embedded pilot-aided structure and the CSC modeling. The CSC-based sparse signal recovery algorithm for channel estimation is presented in Section~\ref{subsec:ce_algorithm}. Section~\ref{sec:simulations} presents the simulation results and their analysis.  Finally, the  paper is concluded in Section~\ref{sec:conc}.   

{\underline{\textit{Notations}}}:  Boldface upper-case, boldface lower-case and lower-case letters  denote the matrices, the vectors and the scalars respectively. For  an $m \times n$  matrix $\bf{A}$, $\text{vec}(\bf{A})$ denotes  the $mn \times 1$ column vector which is obtained  by vertical concatenation of the $n$ columns of $ \bf{A}$. $\mathbb{I}_N$ denotes the identity matrix of size $N \times N$. The all zero matrix of size $n \times m$ is denoted by ${\bf{0}}_{n\times m}$. For a matrix $\bf{A}$, ${\bf{A}}^T$, ${\bf{A}}^{H}$, and ${\bf{A}}^{\dagger}$ represent the transpose, the Hermitian transpose, and the pseudo-inverse of $\bf{A}$ respectively. For any real number $x$, $\lceil x \rceil$ is the  smallest integer that is not smaller than $x$. For any integers $k$ and $N$, the notation $[k]_N$ refers to ($k \mod{N}$). $\circledast$ denotes circular convolution and $\odot$ denotes Hadamard product.  $A=|\mathbb{A}|$ denotes the cardinality of the modulation alphabet $\mathbb{A}$. $\mathbb{C}$ denotes the set of complex numbers.  The notation ${\mathcal{CN}}\left(0, \sigma^2\right)$  denotes a zero-mean complex Gaussian random number with variance $\sigma^2$.  Tx and Rx denote  transmitter and receiver, respectively.
\section{Preliminaries}\label{sec::prelim}
\subsection{OTFS}\label{subsec:otfs_prel}
We discuss the   basic operations involved in  OTFS  without delving into the conceptual details. OTFS  modulation considers an $N\times M$ DD grid ($N$ Doppler bins, $M$ delay bins) for processing  the input and fetching the estimated output data. The delay and the Doppler bins are considered in the horizontal and vertical directions of the DD grid. As  the input and output data are perceived in the DD domain rather than the conventional TF domain, OTFS modulation includes an additional pre-processing block of inverse symplectic finite Fourier transform (ISFFT) and  post-processing block of  SFFT. This structure of OTFS makes it compatible with the  existing  OFDM system. 
 At first, the  OTFS modulator  converts the input data $x[k,l]$ in the DD domain to the  symbols $X[n,m]$  in the TF domain  using ISFFT operation: $X[n,m] =  \frac{1}{MN}\sum_{k=0}^{N-1}\sum_{l=0}^{M-1}x[k,l]e^{j2\pi\left(\frac{nk}{N} - \frac{ml}{M}\right)}$. The TF data is converted to time-domain signal $s(t)$ by applying Heisenberg transform: $s(t) = \sum_{n=0}^{N-1}\sum_{m=0}^{M-1}X[n,m]e^{j2\pi m\triangle f(t-nT)}g_{\text{tx}}(t-nT)$ where $g_{\text{tx}}(t)$  is the transmit basis pulse. The signal $s(t)$ is transmitted through a wireless communication channel whose DD-domain response is $h(\tau,\nu)$. Hence the received signal in time domain is given by $r(t) = \iint h(\tau, \nu)e^{j2\pi\nu(t-\tau)}s(t-\tau)\;d\tau d\nu  +z(t)$, where $z(t)$  is the additive white Gaussian noise (AWGN) signal. At the receiver, by applying Wigner transform,  the time-domain signal is converted to TF domain:  $Y[n,m] =   \int e^{-j2\pi\nu(t-\tau)}g_{\text{rx}}^{\ast}(t-\tau)r(t) \, dt \arrowvert_{\tau=nT, \; \nu=m\triangle f}$, where $g_{\text{rx}}(t)$ is the receive basis pulse. Finally, SFFT  converts the TF signal back  to the DD-domain: $y[k,l] =  \sum_{n=0}^{N-1}\sum_{m=0}^{M-1}Y[n,m]e^{-j2\pi(\frac{nk}{N} - \frac{ml}{M})}$.

The benefit of DD-domain processing is  fully realized if the  pulses $g_{\text{tx}}(t)$  and $g_{\text{rx}}(t)$  satisfy the  so-called \textit{bi-orthogonality} property \cite{ravi_2018_TWC}, under which case, they are called \textit{ideal} pulses.   For ideal pulse-shaping, the input-output relation in the DD domain is given by
\begin{equation}
y[k,l] = \sum_{i=1}^{P}h_{i}x[[k-k_i]_N,[l-l_i]_M] + z[k,l]
\label{eq:otfs_basic}
\end{equation}
where, $k=0,1,\ldots ,N-1$, $l=0,1, \ldots ,M-1$; $P$ is the total number of paths; $h_{i}\sim \mathcal{CN}(0,\frac{1}{P})$, $k_i$, and $l_i$ denote the complex channel gain, integer Doppler  and integer delay tap,  respectively  of the $i^{\text{th}}$ path; and ${z}[k,l]$ is the complex AWGN. For a given channel, $l_{\tau}$ and $k_\nu$ denote the maximum integer delay and integer Doppler tap, respectively. For  maximum delay $\tau_{\max}$ and maximum Doppler $\nu_{\max}$, we must have $\tau_{\max} < \frac{l_\tau}{M\triangle f}$ and $\nu_{\max} < \frac{k_\nu}{NT}$ where, $\triangle f$ is the sub-carrier bandwidth and $T$  is the symbol duration satisfying $\triangle f=1/T$. The input-output relation in (\ref{eq:otfs_basic}) can be compactly expressed as $\mathbf{y}=\mathbf{Hx+z}$, where $\mathbf{H}\in \mathbb{C}^{{NM}\times NM}$; $\mathbf{x}$, $\mathbf{y}$ and $\mathbf{z}$ are the $NM\times 1$ input, output and noise vectors formed by $\text{vec}(\cdot)$ of the corresponding $N\times M$ grids.   Note that ${\bf{H}}$ follows a sparse and circulant-block structure \cite{Raviteja2019a}.  

\subsection{SCMA}\label{subsec:scma_prel}
SCMA is a code-domain NOMA technique in which the available $K$ orthogonal resources (time-slot/frequency-band/code) are shared among  $J$ users ($J>K$) \cite{nikopour2013,taherzdeh2014}. This structure is denoted by $\left(J,K\right)$  SCMA system with an overloading factor of $\lambda=\frac{J}{K}>100\%$. Each user has a  specific codebook having $A$ codeword vectors of length $K$. The codebooks are designed in such a way that each of the $K$ resources is shared by ${d_f}$ users and each  codeword has only $d_v$ non-zero components. As an example, consider a $\left(J=6, K=4\right)$ SCMA system  with $\lambda=150\%$. The set of codebooks can be represented by  $\{\mathbf{C}_1, \mathbf{C}_2, \ldots ,\mathbf{C}_6\}$. For $A=4$, the codebook for the $j^{\text{th}}$  user is given by $\mathbf{C}_j = [\mathbf{c}_{j1} | \mathbf{c}_{j2} | \mathbf{c}_{j3}| \mathbf{c}_{j4}]$ where $\mathbf{c}_{ji} \in \mathbb{C}^{4\times 1}$, $i=1, \ldots, 4$. If the input data of the $j^{\text{th}}$ user is $i$, the corresponding  codeword vector $\mathbf{c}_{ji}$  is selected.  In this way, 6 codeword vectors are identified and  they are simultaneously transmitted over the $4$ available resources. Due to  the sparse structure of the SCMA codewords,  message passing   algorithm  (MPA) can be successfully used for the  data detection \cite{wen_SCMA_2017,discretization}. 
\subsection{OTFS-SCMA}\label{subsec:otfs_scma}
OTFS-SCMA is a code-domain NOMA approach for OTFS  recently proposed in \cite{otfs_scma}.    The key features of  OTFS-SCMA are described in the following. 
\subsubsection{Codeword Allocation Schemes}
The significant difference of OTFS-SCMA from the other related works is that it uses SCMA vector codewords as data symbols instead of QAM symbols.
In \cite{otfs_scma}, three schemes of SCMA codeword allocation are presented.  In \textit{Scheme-1}, the codewords are placed along the Doppler  axis as  $K\times 1$ vectors and hence $N$ should be an integer multiple of $K$.   \textit{Scheme-2} allocates the   codewords as  $1\times K$ vectors along the delay axis with  $M$ being an integer  multiple of $K$. Also, a third scheme, \textit{Scheme-3} is analyzed where the non-zero components of the codewords are swapped in a particular fashion  after  allocating the codewords as per \textit{Scheme-1}   or  \textit{Scheme-2}.   Note that, in all the three schemes,  the overall overloading factor of OTFS-SCMA is the same as that of the underlying basic $\left(J, K\right)$ SCMA system, i.e., $\lambda=J/K$. 
\subsubsection{Downlink and Uplink}
In the downlink scenario, the codewords from   $J$ SCMA encoders are superimposed first,  followed by an  OTFS modulator. The input-output relationship of the ${j}^{\text{th}}$ user is given by 
\begin{equation}
\mathbf{y}_{j} = \mathbf{H}_j\mathbf{x}_{\text{sum}} + \mathbf{z}_j
\label{otfs_scma_dl}
\end{equation}
where $\mathbf{x}_{\text{sum}}$ is  the superimposed input and $\mathbf{H}_j$ is the channel matrix for the  $j^{\text{th}}$ user. For a particular user, the received data is first  passed though an OTFS detector (LMMSE detector  using $\mathbf{H}_j$ ) to resolve the DD interference. The output of the OTFS detector is a noisy version of $\mathbf{x}_{\text{sum}}$.  Finally, an SCMA detector (MPA for AWGN channel) acts upon the noisy $\mathbf{x}_{\text{sum}}$ to detect  the user's data by removing the multi-user interference.
\vspace{-0.2in}
\begin{figure}[!htbp]
\centering
\includegraphics[height=6cm, width=16cm]{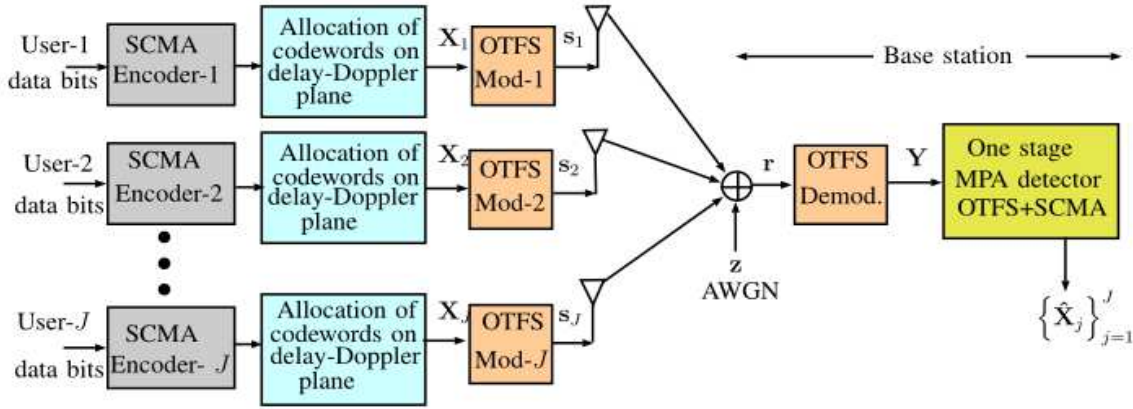}
\caption{\footnotesize{Block diagram of  OTFS-SCMA in uplink.}}
\label{fig:os_uplink}
\end{figure}

Fig.~\ref{fig:os_uplink} depicts an uplink scenario of OTFS-SCMA.  The input-output relation is governed by
\begin{equation}
\mathbf{y} = \sum_{j=1}^{J}\mathbf{H}_j\mathbf{x}_{j} + \mathbf{z} = \mathbf{H}\mathbf{x} + \mathbf{z}
\label{otfs_scma_ul}
\end{equation}
where $\mathbf{H}=\left[{\bf{H}}_1, \ldots, {\bf{H}}_J\right]\in \mathbb{C}^{NM \times JNM}$ and $\mathbf{x}=\left[{\bf{x}}_1^T, \ldots, {\bf{x}}_J^T\right]^T \in \mathbb{C}^{JNM \times 1}$. Observe that $\mathbf{H}$ has the combined effect of both OTFS and SCMA. It is not possible to segregate the DD interaction and the multi-user fusion.   Thus the sequential OTFS and SCMA detection  is not feasible in uplink. A  combined detector for OTFS-SCMA is proposed in \cite{otfs_scma} using MPA, which resolves the DD and the multi-user interference in single stage. 
\subsubsection{Channel Estimation in Downlink}
  The downlink scenario explained above indicates that the multipath channel values are required only for the OTFS detector.  For the SCMA detector,  only the effect of the AWGN channel remains. Hence, the embedded QAM-pilot-based channel estimation technique   with a guard band from \cite{ch_est1} is successfully extended to  OTFS-SCMA in downlink. The results presented in \cite{otfs_scma} demonstrate that a single QAM pilot is sufficient for the channel estimation with a pilot power of $35$ dB so that the BER performance closely follows that of  the perfect CSI case. This strategy is not applicable for OTFS-SCMA in uplink due to the
  very high guard band overhead. Hence, in this paper, we propose an efficient channel estimation
  technique suitable for  uplink scenarios of OTFS-SCMA. 
\subsubsection{Diversity and BER Analysis}
The diversity analysis in \cite{otfs_scma} shows that  OTFS-SCMA can achieve a significantly higher asymptotic diversity order than OTFS-OMA.
The  use of  vector codewords instead of the conventional QAM symbols paves the way for the diversity gain of OTFS-SCMA. 
 \textit{Theorem 1} presented in \cite{otfs_scma} derives the diversity orders of \textit{Scheme-1} and \textit{Scheme-2} for both uplink and downlink. While the diversity gain of \textit{Scheme-1} depends on the  number of distinct mod-$K$ Doppler taps, for \textit{Scheme-2}, it is related to the  number of  distinct  mod-$K$ delay taps. Note that \textit{Scheme-1} and \textit{Scheme-2} perform equally for high values of $N$ and $M$. \textit{Scheme-3} can achieve a higher diversity order by interleaving the non-zero components of the SCMA codewords  based on the  channel information at the transmitter.  The simulation results obtained for the practical EVA channel model \cite{eva}
 and different overloading factors agree with the diversity analysis, which shows that OTFS-SCMA can perform remarkably better than other multi-user OTFS schemes.
 \vspace{-0.1in}
\subsection{Convolutional Sparse Coding}
Convolutional sparse coding (CSC) is a structured coding technique that has found applications in many signal processing problems.  CSC was initially applied to one-dimensional signals \cite{csc_1d} and then extended
to two-dimensional ones \cite{csc_2d,decon_net}.   In CSC, a  signal $\mathbf{y}\in \mathbb{C}^{U\times 1}$ is represented as the sum of $v$ convolutions: 
\begin{equation}
\min_{\mathbf{x}_i}{f(\mathbf{x}_i)} \text{ s.t } \sum_{i=1}^{v}\mathbf{d}_i\circledast{\mathbf{x}_i}=\mathbf{y}
\label{eq:csc_1}
\end{equation} 
where $f(\cdot)$ is usually $\ell_1$-norm, $\{\mathbf{d}_{i}\}_{i=1}^{v}  \in {\mathbb{C}}^{u \times 1}$ are support filters ($u\ll U$),  and $\{\mathbf{x}_{i}\}_{i=1}^{v} \in {\mathbb{C}}^{U\times 1}$ are vectors of varying sparsity. The CSC problem of (\ref{eq:csc_1}) can be compactly represented as \cite{fast_csc,eff_csc}  
\begin{equation}
\min_{\tilde{\mathbf{x}}}{f(\tilde{\mathbf{x}})} \text{ s.t } \tilde{\mathbf{D}}\tilde{\mathbf{x}}={\mathbf{y}}
\label{eq:csc_2}
\end{equation} 
where $\tilde{\mathbf{D}}=\left[\tilde{\mathbf{D}}_1 \;  \tilde{\mathbf{D}}_2\; \ldots\; \tilde{\mathbf{D}}_v\right]$, $\tilde{\mathbf{D}}_i \in {\mathbb{C}}^{U \times U}$ is the circulant matrix formed by  $\mathbf{d}_i$ and its $U-1$ circularly shifted vectors and $\tilde{\mathbf{x}}=\left[{\mathbf{x}}_1^{T} \; {\mathbf{x}}_2^{T}\; \ldots \; {\mathbf{x}}_v^{T} \right]^{T}$.
Observe that  a block circulant shift structure  is present in the coding dictionary  $\tilde{\mathbf{D}}$. Moreover, the sparsity of the signal $\tilde{\mathbf{x}}$ makes it possible to use sparse signal recovery algorithms.

Recently, less complex CSC algorithms are developed based on  \textit{working locally; thinking globally} \cite{csc_slice}. Under this local paradigm,  the global dictionary  is broken down into  smaller  local dictionaries $\mathbf{D}_{L}$ of dimension $u\times v$, where $\mathbf{D}_{L} = \left[\mathbf{d}_1 \; \mathbf{d}_2  \; \ldots \; \mathbf{d}_v  \right]$. Through a simple permutation of its columns, $\tilde{\mathbf{D}}$ can be represented as the concatenation of circularly shifted versions of $\mathbf{D}_L$. Fig.~\ref{fig:csc_1}  depicts the process of converting the global problem  into a local  one.

\begin{figure}[!htbp]
	\centering
	\includegraphics[height=6cm, width=16cm]{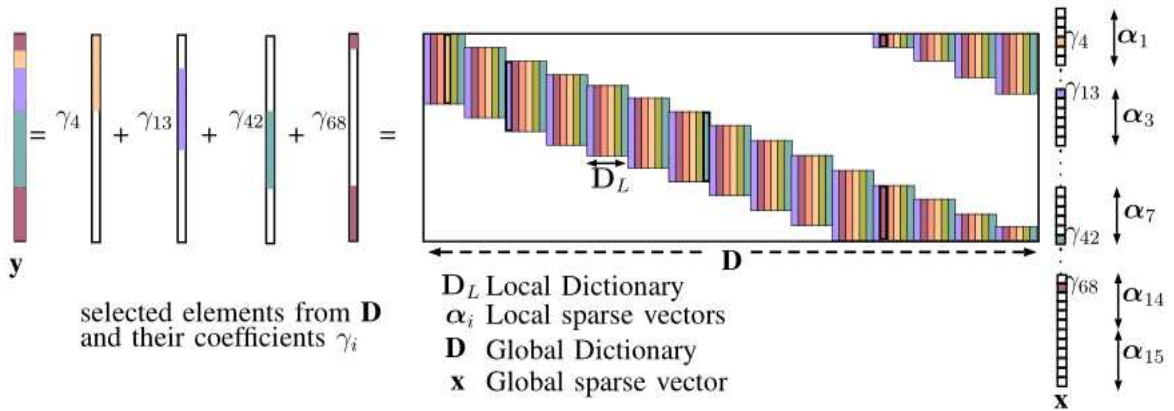}
	\caption{Dictionary structure of convolutional sparse coding.}
	\label{fig:csc_1}
\end{figure}

The CSC problem from  (\ref{eq:csc_2})  can now be written as 
\begin{equation}
\min_{\boldsymbol{\alpha}_i}{f(\boldsymbol{\alpha}_i)} \text{ s.t } \sum_{i=1}^{U}\mathbf{R}_i^{T}\mathbf{D}_L\boldsymbol{\alpha}_i=\mathbf{y}
\label{eq:csc_slice}
\end{equation}
where $\boldsymbol{\alpha}_i$ is the local $v\times 1$ sparse vector and $\mathbf{R}_i=\left[\mathbf{0}_{u\times (i-1)} \ \mathbb{I}_{u} \  \mathbf{0}_{u\times (U-i-u+1)}\right]$ is a  $u\times U$ matrix to extract a $u\times 1$ patch from $\mathbf{y}$. The CSC applications discussed in \cite{LoBCoD} and \cite{csc_greedy} establish that the sparse signal recovery algorithms based on convex relaxation (e.g. gradient descent) and greedy approach (e.g. matching pursuit, OMP, stage-wise OMP) are efficient when working with local dictionary also. The  use of   local dictionary helps to consider the global CSC of $\mathbf{y}$ in terms of independent CSC problems of smaller vectors of length $u$. It is interesting to note that the input-output relation for 
the proposed method with the DD grid having embedded pilot vectors can be modeled similarly to (\ref{eq:csc_1}), and its simplified form in Fig.~\ref{fig:csc_1}. This observation motivated us to develop a channel estimation technique for OTFS-SCMA using the
CSC approach, which is explained thoroughly in the next section.

\section{Formulation of Channel Estimation as a CSC Problem}\label{sec:proposed}
This section comprehensively describes the CSC-based  channel estimation for OTFS-SCMA, which is  organized as follows.  In Section~\ref{subsec:data_pilot}, we present the arrangement of the  data and the pilot vectors in the transmit DD grid and analyze the corresponding received grid. Section~\ref{subsec:ce_csc} formulates  the channel estimation task  as a CSC problem, based on the data and the pilot arrangement.

We initially develop the channel estimation technique for the  ideal pulse shaping in which the   input-output relation follows a simple 2D circular convolution as shown in (\ref{eq:otfs_basic}). Later,  the proposed method is  extended to the rectangular pulse shaping case. 

\subsection{Pilot and Data Arrangement in Delay-Doppler Grid}\label{subsec:data_pilot}
Although the proposed method applies to any codeword allocation  schemes, for simplicity, we consider  \textit{Scheme-1} to describe the channel estimation technique. 
Let  ${\mathbf{X}}_u \in{\mathbb{C}}^{N\times M}$  denote the $u^{\text{th}}$  user's input symbol matrix which is placed on the respective $N\times M$  Tx DD grid.  If ${\mathbf{x}}_{u,{l}}\in{\mathbb{C}}^{N\times 1}$ denotes the column vector for delay tap $l$,  the input symbol matrix can be written as ${\mathbf{X}}_u =\left[{\mathbf{x}}_{u,{0}}  \ {\mathbf{x}}_{u,{1}} \ \ldots \ {\mathbf{x}}_{u,{M-1}} \right]$. The  element in the location $[k,l]$ of  ${\mathbf{X}}_u$ is denoted by ${x}_u\left[k,l\right]$  with $k=0,1, \ldots, N-1$  and $l=0,1, \ldots, M-1$. 
We consider column vectors of specific length for pilots. Suppose the pilot vector of the $u^{\text{th}}$ user is denoted by $\mathbf{p}_u=\left[p_{u,0}\  p_{u,1} \ \ldots \ p_{u,L_p-1}\right]^T$ and is placed at the delay index $\bar{l}$   such that $ \bar l +l_\tau=M-1$ where,  $l_{\tau}$ is the maximum integer delay tap. Thus a just adequate guard band from the data part is maintained  to minimize the  overhead. 
The total number of data symbols that can be transmitted over the DD grid is  ${N_{\text{data}}}=\frac{N}{K}(\bar{l}-l_\tau)$.  The $n^{\text{th}}$  SCMA codeword transmitted by the $u^{\text{th}}$  user is given by   ${\bf{c}}_n^u= \left[{{c}}_{n,0}^u \ {{c}}_{n,1}^u \ \ldots \ {{c}}_{n,K-1}^u\right]^T$, $n=1, 2, \ldots, {N_{\text{data}}}$.   
The symbols in the Tx grid  of the $u^{\text{th}}$ user can be described as 
\begin{equation}
{x}_u[k,l]=\begin{cases}
{c}_{n,i}^u &0\leq l<\bar{l}-l_\tau, \ k=[(n-1)K+i]_N, \;\;
\text{(data symbol)}\\
p_{u,i}&l=\bar{l}\ , \ k=i<L_p \;\;\;\;\;\;\;\;\;\;\;\;\;\;\;\;\;\;\;\;\;\;\;\;\;\;\;\; \;\; \text{(pilot symbol)}\\
0&l=\bar{l}\ , \ k>L_p-1 \;\;\;\;\;\;\;\;\;\;\;\;\;\;\;\;\;\;\;\;\;\;\;\;\;\;\;\;\;\; \text{(guard band)}\\
0& \text{otherwise.} \;\;\;\;\;\;\;\;\;\;\;\;\;\;\;\;\;\;\;\;\;\;\;\;\;\;\;\;\;\;\;\;\;\;\;\;\;\; \;\;\;\;\;\text{(guard band)} 
\end{cases}
\label{eq:prop_tx}
\end{equation}
\begin{exm}\label{eg:params}Consider  an OTFS-SCMA system with $M=5$, $N=8$, $J=6$, $K=4$, $d_v=2$, $l_{\tau}=k_{\nu}=1$, and $L_p=4$. The codewords follow the sparsity pattern as per the following factor matrix:
		\begin{equation}
		\mathbf{F} = 
		\begin{bmatrix}
		1&0&1&0&1&0\\
		0&1&1&0&0&1\\
		1&0&0&1&0&1\\
		0&1&0&1&1&0
		\end{bmatrix}.
		\label{eq:fact_graph}
		\end{equation}
		The pilot vector is placed at the delay tap $\bar{l}=3$.
		This example will be revisited multiple times to illustrate various concepts and procedures. 
\end{exm}

\begin{figure}[!htbp]
	\subfloat[Tx DD grids of User 1 ($\tilde{\mathbf{X}}_1$) and User 6 ($\tilde{\mathbf{X}}_6$)]{\includegraphics[scale=1]{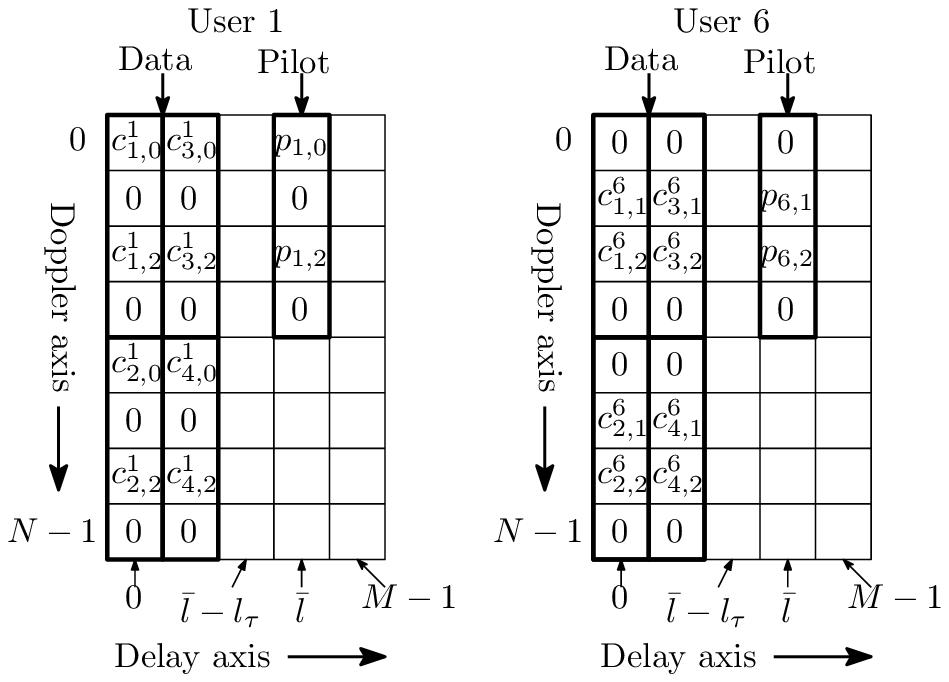}
		\label{fig:prop_tx}
	}
	\subfloat[Rx DD grid ($\mathbf{Y}$)]{\includegraphics[scale=1]{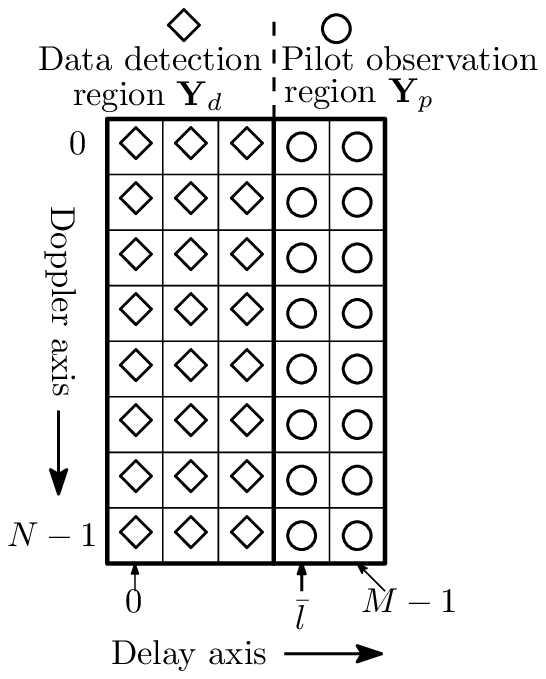}
		\label{fig:prop_rx}}
	\caption{Pilot-data arrangement in DD domain for Example~\ref{eg:params} (For convenience,  pilot vectors of length 4 are shown.  The selection of the length of pilot vectors is discussed in Section~\ref{subsec:optimal_length}.)  }
	\label{fig:tx_rx_grid}
\end{figure}

Fig.~\subref*{fig:prop_tx} shows the embedded arrangement of  the data and the pilot vectors over the  Tx DD grid   for Example~\ref{eg:params}.  The codewords and the pilot vectors follow the sparsity pattern as per (\ref{eq:fact_graph}).   The Tx grids for only User 1 and  User 6 are shown. 
The received symbols $\mathbf{Y}\in \mathbb{C}^{N\times M}$ is the collection of $M$ columns as $[\mathbf{y}_0 \ \mathbf{y}_1 \ldots \mathbf{y}_{M-1}]$.  
As a result of the sufficient guard band there are two non-overlapping regions:  (1)~data detection region $\mathbf{Y}_d\in \mathbb{C}^{N\times \bar{l}}$ formed by the columns $\mathbf{y}_0, \mathbf{y}_1, \cdots, \mathbf{y}_{\bar{l}-1}$ and  (2)~pilot observation region $\mathbf{Y}_p\in \mathbb{C}^{N\times (l_\tau +1)}$ formed by the columns $\mathbf{y}_{\bar{l}}, \mathbf{y}_{\bar{l}+1}, \cdots, \mathbf{y}_{\bar{l}+l_\tau}$ . 
For Example~\ref{eg:params},  the Rx grid  with multi-user and DD interference is  shown  in Fig.~\subref*{fig:prop_rx}. 
\vspace{-0.1in}
\subsection{CSC Model for Channel Estimation}\label{subsec:ce_csc}
The channel estimation task can be carried out by considering it as a CSC problem.  For that,  the input-output relation of OTFS-SCMA is analyzed for the single-user case first and then for the multi-user one.\\
\textbullet \hspace{5pt}\textit{{Single user case:}} Let ${h}_{u}{[k^\prime,l^\prime]}$ denote the channel coefficient of the $u^{\text{th}}$ user for the path having delay tap  $l^\prime$ and Doppler tap  $k^\prime$ where $0\le l^\prime \le M-1$  and $ 0\le k^\prime \le N-1$. Ignoring the AWGN at the BS, the input-output relation for the $u^{\text{th}}$  user from  (\ref{eq:otfs_basic})  can  be written as
\begin{equation}
	y_u[k,l] = \sum_{ k^\prime=0}^{N-1}\sum_{ l^\prime=0}^{M-1}{h}_{u}{[k^\prime,l^\prime]}{x}_u\left[\left[k-k^\prime\right]_N,[l-l^\prime]_M \right].
	\label{eq:otfs_basic1}
\end{equation}
The number $P$ of multipaths  is the number of $(k^\prime,l^\prime)$ pairs for which ${h}_{u}{[k^\prime,l^\prime]}\neq 0$. 
Let  ${\mathbf{y}}_{u,{l}}=\left[y_u[0,l] \ y_u[1,l] \ \ldots \ y_u[N-1,l]\right]^T$ and ${\mathbf{h}}_{u,m} =\left[{h}_{u}{[0,m]} \ {h}_{u}{[1,m]} \ \ldots \ {h}_{u}{[N-1,m]}\right]^T$. Then, the input-output relation in (\ref{eq:otfs_basic1}) can be written in the vector form as
\begin{equation}
	\mathbf{y}_{u,{l}} = \sum_{l^{\prime}=0}^{M-1}{{\mathbf{x}}_{u,{[l-l^\prime]_M}}\circledast {\mathbf{h}}_{u,{l^\prime}}}\ ,\ \ \text{for } l=0,1,\ldots M-1 .
	\label{eq:io_conv}
\end{equation}
For the purpose of channel estimation, we apply the following  limits  for  $l^\prime$  and $l$ in (\ref{eq:io_conv}): 
\begin{itemize}
	\item $0\leq l^\prime \leq l_\tau$, since ${\mathbf{h}}_{u,{l^\prime}}=\mathbf{0}$ for $l^\prime > l_\tau$.
	\item $\bar{l}\leq l \leq M-1$, since we are interested in the pilot observation region $\bf{Y}_p$ only.
\end{itemize}
Thus (\ref{eq:io_conv})  is simplified as
	\begin{equation}
	\mathbf{y}_{u,{l}} = \sum_{l^{\prime}=0}^{l_\tau}{{\mathbf{x}}_{u,{[l-l^\prime]_M}}\circledast {\mathbf{h}}_{u,{l^\prime}}}\ ,\ \ \text{for } l=\bar{l},\ \bar{l}+1,\ldots M-1 .
	\label{eq:io_conv_ce}
\end{equation}
The matrix structure of (\ref{eq:io_conv_ce}) is given below:
\begin{equation}\left[
		\begin{array}{c}
			\mathbf{y}_{u,{\bar{l}}}\\
			\mathbf{y}_{u,{\bar{l}+1}}\\
			\vdots\\
			\mathbf{y}_{u,{M-1}}\\
		\end{array}\right]
		=
		\left[\begin{array}{cccc}
			\mathbf{P}_{u}&\mathbf{0}_{\tiny{N\times N}}&\ldots&\mathbf{0}_{\tiny{N\times N}}\\
			\mathbf{0}_{\tiny{N\times N}}&\mathbf{P}_{u}&\ldots&\mathbf{0}_{\tiny{N\times N}}\\
			\vdots&\vdots&\ddots&\vdots\\
			\mathbf{0}_{\tiny{N\times N}}&\mathbf{0}_{\tiny{N\times N}}&\ldots&\mathbf{P}_{u}\\
		\end{array}\right]
		\left[\begin{array}{c}
			{\mathbf{h}}_{u,0}\\
		{\mathbf{h}}_{u,1}\\
			\vdots\\
			{\mathbf{h}}_{u,{l_\tau}}\\
		\end{array}\right]
	\label{eq:su_csc}
\end{equation}
where $\mathbf{P}_{u} \in {\mathbb{C}}^{N \times N}$ is the circulant matrix formed by the pilot vector ${\bf{p}}_u$ as explained in the following. For any given $\mathbf{p}_{u}$, let $\mathbf{p}_{u}^{(k)}$ denote the   transpose of the $k^{\text{th}}$ forward circular shift of $[{\mathbf{p}^T_{u}} \ \mathbf{0}_{1\times {(N-L_p)}}]$, corresponding to the Doppler tap $k$. Then $\mathbf{P}_{u}$ is given by $\left[\mathbf{p}_{u}^{(0)} \ \mathbf{p}_{u}^{(1)} \ \ldots  \ \mathbf{p}_{u}^{(N-1)}\right]$. Considering the block diagonal structure in (\ref{eq:su_csc}), we have $l_\tau+1$ independent equations as 
\begin{equation}
	\mathbf{y}_{u,l} = \mathbf{P}_{u}\mathbf{h}_{u,{l-\bar{l}}} \ \ \text{for } l=\bar{l}, \ \bar{l}+1,\ldots , M-1 .
	\label{eq:ce_simple}
\end{equation}  
\begin{rem}
	For the pilot and data arrangement   in Section~\ref{subsec:data_pilot},   (\ref{eq:ce_simple}) is applicable to all users  as sufficient guard band is reserved considering  the maximum delay spread of the channels.  
\end{rem}

\textbullet \hspace{5pt} \textit{{Multi user case:}} Now, we extend (\ref{eq:ce_simple}) to the OTFS-SCMA uplink with $J$  users. The pilot observation signal at the BS is given by
\begin{align}
	\mathbf{y}_{l} &=  \mathbf{y}_{1,l}+ \mathbf{y}_{2,l}+ \cdots +\mathbf{y}_{J,l} =\mathbf{P}_{1}\mathbf{h}_{1,{l-\bar{l}}} + \mathbf{P}_{2}\mathbf{h}_{2,{l-\bar{l}}}+ \ldots +\mathbf{P}_{J}\mathbf{h}_{J,{l-\bar{l}}} \label{eq:io_relation_uplink}\\
	&={[\mathbf{P}_{1}\; \mathbf{P}_{2} \ldots \mathbf{P}_{J}]}{\left[\begin{array}{c}
			\mathbf{h}_{1,{l-\bar{l}}}\\
			\mathbf{h}_{2,{l-\bar{l}}}\\
			\vdots\\
			\mathbf{h}_{J,{l-\bar{l}}}
		\end{array}\right] } 
=\sum_{u=1}^{J}\left({\sum_{k^\prime = 0}^{N-1}}{\mathbf{p}^{(k^\prime)}_{u}h_{u}{[k^\prime,l-\bar{l}]}}\right)  \nonumber \\
& =\sum_{k^\prime = 0}^{N-1}\left({\sum_{u=1}^{J}}{\mathbf{p}^{(k^\prime)}_{u}h_{u}{[k^\prime,l-\bar{l}]}}\right)\ \ \;\;\;\text{for } l=\bar{l}, \ \bar{l}+1,\ldots , M-1
\label{eq:rearrange}
\end{align}

\begin{figure}[h]
	\includegraphics[scale=1]{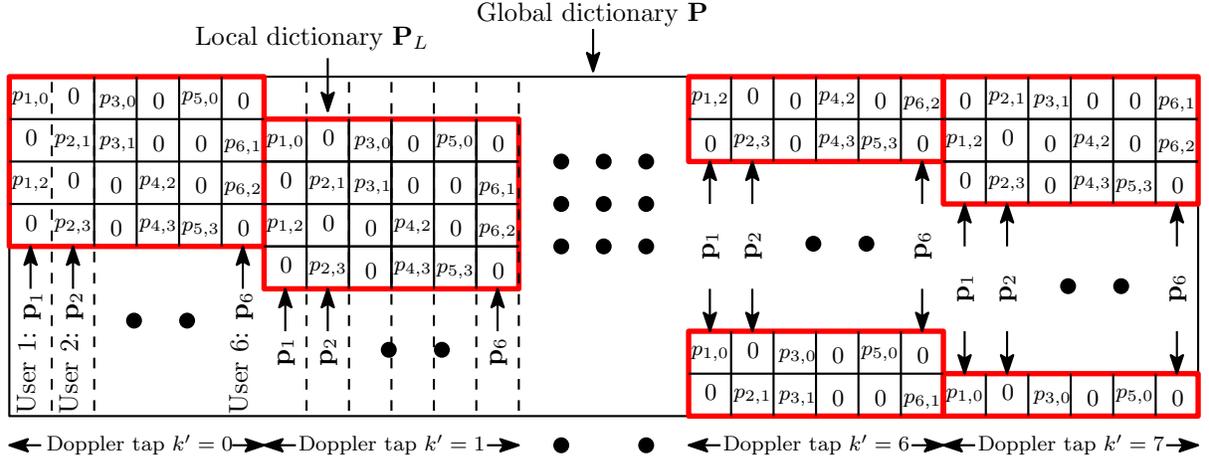}
	\captionof{figure}{Structure of   global dictionary ${\bf{P}}$ and local dictionary ${\bf{P}}_L$ for Example~\ref{eg:params}.}
	\label{fig:ce_csc_1}
\end{figure}

Expressing (\ref{eq:rearrange}) in the matrix form, we get 
\begin{equation}
	\mathbf{y}_{l}=\mathbf{P} {\bf{h}}_{l-\bar{l}}\ \ \;\;\;\text{for } l=\bar{l}, \ \bar{l}+1,\ldots , M-1
	\label{eq:rearrange_mat}
\end{equation}
where $\mathbf{P}  \in {\mathbb{C}}^{N \times JN}$  is given by
 \[{{\bf{P}}} = \left[ {\underbrace {{\bf{p}}_1^{(0)}\;{\bf{p}}_2^{(0)} \ldots \;{\bf{p}}_J^{(0)}}_{{\rm{Doppler}}\;\;{\rm{tap}}\;\;k^\prime = 0}\;\underbrace {{\bf{p}}_1^{(1)}\;{\bf{p}}_2^{(1)}\; \ldots \;{\bf{p}}_J^{(1)}}_{{\rm{Doppler}}\;\;{\rm{tap}}\;\;k^\prime = 1} \ldots \;\underbrace {{\bf{p}}_1^{(N - 1)}\;\;{\bf{p}}_2^{(N - 1)}\; \ldots \;{\bf{p}}_J^{(N - 1)}}_{{\rm{Doppler}}\;\;{\rm{tap}}\;\;k^\prime = N - 1}} \right]\]
and $\mathbf{{h}}_{l^\prime}$  is given by
\[
 \begin{aligned}
\mathbf{{h}}_{l^\prime} &= 
 	 \left[ {{{\boldsymbol{\alpha }}_{(0,l^\prime)}^T 
 		 \ }\cdots \ {{\boldsymbol{\alpha }}_{(k_\nu,l^\prime)}^T} \ {{{\bf{0}}}_{1 \times J\left( {N - {(2k_\nu+1)}}\right)}}  \ {{{\boldsymbol{\alpha }}_{(N-k_\nu,l^\prime)}^T} \cdots {{\boldsymbol{\alpha }}_{(N-1,l^\prime)}^T} }}\right]^T
 \end{aligned}
 \]
 where ${{\boldsymbol{\alpha }}_{\left( {k^\prime,l^\prime} \right)}} ={\left[ {{h_1}\left[ {{k^\prime },{l^\prime }} \right],{h_2}\left[ {{k^\prime },{l^\prime }} \right], \ldots ,{h_J}\left[ {{k^\prime },{l^\prime }} \right]} \right]^T}$.  Note that we have put 0s in $\mathbf{{h}}_{l^\prime}$ by taking into account both positive and negative Doppler values.
 
 The local dictionary associated with ${\bf{P}}$  is given by
 ${{\bf{P}}_L} = \left[ {{{\bf{p}}_1}\;\;{{\bf{p}}_2}\; \cdots \;{{\bf{p}}_J}} \right]$.  
 Fig.~\ref{fig:ce_csc_1} shows the structures of the global dictionary ${\bf{P}}$  and the local dictionary ${\bf{P}}_L$  for Example~\ref{eg:params}.
Clearly (\ref{eq:rearrange_mat}) can be considered as a CSC problem similar to that of (\ref{eq:csc_2}) and (\ref{eq:csc_slice}). The channel estimation is now formulated as
\begin{equation}
	\begin{rcases*}
		\mathop {\min }\limits_{{\mathbf{{h}}}_{l-\bar{l}}}{f(\mathbf{{h}}_{l-\bar{l}})} \text{ s.t } \mathbf{P}\mathbf{{h}}_{l-\bar{l}}=\mathbf{y}_{l} \;\;\;\;\;\;\; \text{ or } \label{eq:ce_conv1} & \\
		\mathop {\min }\limits_{\boldsymbol{\alpha}_{(k,l-\bar{l})}}{f(\boldsymbol{\alpha}_{(k,l-\bar{l})})} \text{ s.t } \sum_{k \in {\cal{K}}}\mathbf{R}_k^{T}\mathbf{P}_L\boldsymbol{\alpha}_{(k,l-\bar{l})}=\mathbf{y}_{l}\; & 
	\end{rcases*}  \stackanchor{l=\bar{l}, \ \bar{l}+1,\ldots , M-1 \text{ and }}
{{\cal{K}}=\{0,\ldots,k_\nu,N-k_\nu,\ldots,N-1\}} 
\end{equation}
where  $\mathbf{R}_k$ is the same patch extraction matrix as in (\ref{eq:csc_slice}).  
\section{Algorithm for channel estimation based on CSC}\label{subsec:ce_algorithm}
In this section, we propose an efficient sparse signal recovery algorithm for the CSC-based approach of channel estimation problem given in (\ref{eq:ce_conv1}). We explore the features of sparse pilot vector design  to ensure fast convergence and improved performance.   Algorithm~\ref{algo:ce_alg} shows the steps for the CSC-based channel estimation technique.

\subsection{Low Complexity and Fast Convergence Sparse Signal Recovery Algorithm}
The sparse signal recovery algorithms based on greedy approach are extensively used for channel estimation problems. Algorithm~\ref{algo:ce_alg} is designed based on  SP with a modified method for obtaining the initial estimates to achieve faster convergence.  
\begin{algorithm}[!htbp]
	\captionof{algocf}{\small SP algorithm with modified initial estimate}
	\label{algo:ce_alg}
	\linespread{1}
	\small
	\SetKwData{Left}{left}
	\SetKwData{This}{this}
	\SetKwData{Up}{up}
	\SetKwFunction{Union}{Union}
	\SetKwFunction{FindCompress}{FindCompress}
	\SetKwInOut{Input}{input}
	\SetKwInOut{Output}{output}
	\Input{$\mathbf{Y}_p$: pilot observation region;
		$\mathbf{P}_L$: dictionary of pilots; $J$: total number of users; $k_\nu$: maximum integer Doppler tap; $\mathcal{T}$: Threshold for dominant path detection;
	}
	Obtain  $\mathbf{P}$ using the input 	$\mathbf{P}_L$\\ 
	Limit the search range of Doppler to the maximum spread of  Doppler taps of channel\;
	\ForEach{delay tap $l$ in the pilot observation region}{
		Thresholding: Check whether $l$ is a dominant delay path based on $||\mathbf{y}_l||_2 >{\cal{T}}$. If $l$ is not a dominant path then skip the following part and go to line 3 for  the  next delay tap $l+1$ \;
		\textit{Channel estimation:} 
		
		\If{$||\mathbf{y}_l||_{2}>\mathcal{T}$}{
			Initialization: iteration $i=0$\;
			Initial estimates: Estimate one dominant Doppler tap per user and the residue as in  {\bf{(A1)-(A5)}}\;
			\Repeat{stopping criteria, as in (\ref{eq:stop_crit})}{
				$i=i+1$\;
				Find the $J$ columns of $\mathbf{P}$ with high correlation with the residue as in ({\bf{B1}})\;
				Update the current temporary list $\tilde{S}^{(i)}$ of selected columns  by adding these $J$ columns  to the previous list $S^{(i-1)}$ as in ({\bf{B2)}} \;
				Find the pseudoinverse solution for the $2J$ selected columns of $\mathbf{P}$ as in ({\bf{B3}})\;
				Find ${S}^{(i)}$ and refine the channel estimates finding  the $J$ selected columns of significant coefficients as in {\bf{(B4)}} and {{\bf{(B5)}}}\;
				Update the residue using the refined channel estimates as in (\ref{eq:res_update}) of {\bf{(B6)}}\;
				
			}
			
		} 
	}
	\Output{Channel estimates of all users:  $\hat{\mathbf{{h}}}_{l-\bar{l}}$  for $l=\bar{l}, \bar{l}+1, \ldots , M-1$.}
\end{algorithm}
The steps of Algorithm~\ref{algo:ce_alg} are discussed  next, highlighting the factors crucial for  effective channel estimation.\\
\textbf{Obtaining the global dictionary using the local dictionary:} From (\ref{eq:ce_conv1}), it is observed that the sensing matrix  $\mathbf{P}$ remains the same for all the delay taps and is used as the global dictionary throughout  Algorithm~\ref{algo:ce_alg}. Hence we initially form  $\mathbf{P}$, which is simply the block circulant version of the local dictionary input $\mathbf{P}_L$ as depicted in Fig.~\ref{fig:csc_1}  for Example~\ref{eg:params}.\\
\textbf{Shrinking the global dictionary:} Although $\mathbf{P}$ is formed considering the complete Doppler taps from $0$ to $N-1$, the search range can be limited to the maximum spread of the integer Doppler taps $2k_\nu$ of the channel. We shrink  the dimension of $\mathbf{P}$  from $JN$ to $J(2k_\nu+1)$ columns (line~2 of Algorithm~\ref{algo:ce_alg}). This shrinkage in the dimension of the  sensing matrix reduces the complexity of the sparse signal recovery algorithm. 
\\
\textbf{Thresholding of pilot observation region:} The pilot observation region accommodates the maximum delay tap $l_\tau$ of the  channel but not every delay tap starting sequentially from 0 to $l_\tau$  is present. Only a few dominant delay taps will be present in the channel. If $M_{\text{dom}}$ is  the number of dominant delay taps in  the channel, then in (\ref{eq:ce_conv1}),  we need to apply sparse signal recovery algorithm only to $M_{\text{dom}}$ equations ($M_{\text{dom}}\ll M$). A  thresholding of  $||\mathbf{y}_{l}||_{2}>\mathcal{T}$  can easily find out these $M_{\text{dom}}$ equations, since all other $\mathbf{y}_{l}$ have only the noise component.\\    
\textbf{Obtaining the initial estimates sequentially:} Instead of selecting the initial $J$ elements from the overall set of $JN$ elements, we can strategically choose the initial values based on a group-by-group greedy approach, making the convergence faster (line~8 of Algorithm~\ref{algo:ce_alg}).  There are $J$ groups of $N$ channel coefficients  corresponding to  $J$ users and $N$  Doppler taps of every user at a fixed delay tap. We select one dominant Doppler tap  per user   as given below:
\begin{enumerate}[start=1,label={(\bfseries A\arabic*):}]
	\item Initialize $S^{(0)}=[\ ]$;
	\item For each $j^{\text{th}}$ user, form $\mathbf{P}_{j}$ collecting all the columns corresponding to that user from $\mathbf{P}$.
	\item From each $\mathbf{P}_{j}$, find one index $v_j$ corresponding to the largest entry in the vector $|\mathbf{P}^{H}_{j}\mathbf{y}_{l}|$ and set $S^{(0)}=S^{(0)}\cup v_j$.
	\item After $J$ iterations of ({\bf{A2}}) and ({\bf{A3}}), collect the columns of $\mathbf{P}$ corresponding to those indices in $\mathbf{P}_{\left(S^{(0)}\right)}$.
	\item The residue  to be used in the first iteration is obtained using the initial estimates: \\$\mathbf{r}^{(0)}=\mathbf{y}_{l} - \mathbf{P}_{_{\left(S^{(0)}\right)}}\mathbf{P}_{_{\left(S^{(0)}\right)}}^{\dagger}\mathbf{y}_{l}.$
\end{enumerate}
\textbf{Iterative steps  of the algorithm}: After setting the initial estimates and the residue, the following steps are carried out iteratively from $i=1$ till the stopping criteria are not satisfied.

\begin{enumerate}[start=1,label={(\bfseries B\arabic*):}]
	\item Obtain the $J$ elements which are highly correlated to the residue of the previous iteration, i.e., the $J$ maximum correlation coefficients in the vector $|\mathbf{P}^{H}\mathbf{r}^{(i-1)}|$, and collect these $J$ column indices in $\mathbf{s}^{(i)}$.
	\item Form   the list of $2J$  elements  by including the new $J$ elements of  the current iteration to the list of the previous iteration: \quad $\tilde{S}^{(i)}= S^{(i-1)}\cup \mathbf{s}^{(i)}$.
	\item Find the nearest solution for $2J$ elements using the psuedoinverse method: 
	\begin{equation}\mathbf{h}_{_{\left(\tilde{S}^{(i)}\right)}}= \mathbf{P}_{_{\left(\tilde{S}^{(i)}\right)}}^{\dagger}\mathbf{y}_{l}.\label{eq:pseudo_2j}
	\end{equation}   
	\item Find the  $J$ maximum values in the vector $\left|\mathbf{h}_{_{\left(\tilde{S}^{(i)}\right)}}\right|$ and put those $J$ indices in the selected list ${S}^{(i)}$ of $J$ indices. Collect the columns of $\mathbf{P}$ corresponding to ${S}^{(i)}$ in $\mathbf{P}_{_{\left({S}^{(i)}\right)}}$.
	\item  Refine the $J$ channel estimates as: \hspace{0.2in} $\mathbf{h}_e=\mathbf{P}_{_{\left({S}^{(i)}\right)}}^{\dagger}\mathbf{y}_{l}$.
	\item Using the $J$ estimated coefficients, update  the residue for the next iteration:
	\begin{equation} 
		\mathbf{r}^{(i)}=\mathbf{y}_{l}-\mathbf{P}_{_{\left(S^{(i)}\right)}}\mathbf{h}_e .\label{eq:res_update}
	\end{equation}
\end{enumerate}

\textbf{Stopping criteria:} The algorithm is considered to converge to a correct solution when the  list of the  selected indices remains the same over two consecutive iterations and the norm of the residue increases for the next iteration i.e., the iteration stops when $S^{(i)}=S^{(i-1)}$  and

\begin{equation}||\mathbf{r}^{(i)}||_{2}>||\mathbf{r}^{(i-1)}||_{2}. \label{eq:stop_crit}
\end{equation}

The estimated channel coefficients $\hat{\mathbf{h}}_{l-\bar{l}}$ of the current delay tap are set to the refined estimate $\mathbf{h}_e$  of the previous iteration $i-1$.\\
The features of Algorithm~\ref{algo:ce_alg} that result in fast convergence are highlighted next.

\begin{itemize}
\item {{\textit{Sparse pilot vector and sequential initial estimates}}}: The pilot vector of each user follows the same sparse structure of the user's codewords.  SCMA codebook is designed such that only $d_f$ users share a resource. The similar sparse structure for the pilot vectors favors the sparse signal recovery since those columns of $\mathbf{P}$ will be highly correlated with $\mathbf{y}_{l}$. Also, it ensures that all the $d_v$ non-zero components of all pilot vectors are available for channel estimation. The sparseness of the pilot vectors of the multiple users provides minimal
interference amongst them. In addition, the sequential selection of the initial estimates offers uniform treatment to all users  without prioritizing any user.
\item {{\textit{Dealing with unknown sparsity}}}: The knowledge of sparsity is a prerequisite for SP algorithm. 
The dictionary $\mathbf{P} \in {\mathbb{C}}^{N\times JN}$  is used to recover the channel coefficients at the all  Doppler taps of all users for a single delay tap as shown in  (\ref{eq:ce_conv1}).
The sparse signal recovery in (\ref{eq:ce_conv1}) aims to recover $J$ components from the available $JN$ components of $\mathbf{P}$. i.e., we have a fixed sparsity of $\frac{1}{N}$ for each  of the $M_{\text{dom}}$ sparse recovery problems under consideration. Thus, the unknown sparsity in the number of paths is  translated to  a known sparsity of the number $J$ of users.
\end{itemize}
\subsection{Pilot Pattern Design}
\subsubsection{Length of Pilot Vector}\label{subsec:optimal_length}
Sparse signals can be successfully  recovered   if the dictionary elements follow specific criteria \cite{cs_eldar}.  Considering these conditions and the requirements for OTFS-SCMA, we  state a condition for the length of the pilot vector in the following lemma.
\begin{lemma}\label{rem:optimal_Lp}
	For successful channel estimation using the proposed method, the  length $L_{p}$ of pilot vector must satisfy the following  condition:
	\vspace*{-.1in}
	\begin{equation}
		L_p \geq \max\left\{2J,  \left\lceil c J\log(J(2k_\nu + 1))\right\rceil -2k_\nu, k_\nu+1\right\}  \;\;\;  \text{ with } \left[L_p\right]_K=0
		\label{eq:Lp_condtion}
	\end{equation}
	where $c$  is a constant  satisfying $1<c \leq 2$.
\end{lemma}
\begin{proof}
	The proof is given in Appendix~\ref{app:optimal_lp}.
\end{proof}
\subsubsection{Pilot Vector Sequence}\label{subsec:pilot_seq}
The design of the pilot sequence plays a crucial role in successful
channel estimation. For compressive-sensing-based techniques, the primary factor of concern is the cross-correlation among the columns of sensing matrix (dictionary) $\mathbf{P}$. In that regard, one important figure of merit is the mutual coherence of $\mathbf{P}$ which is defined as $\mu(\mathbf{P})=\max_{i\neq j}|\mathbf{p}_{(i)}^{H}\mathbf{p}_{(j)}|$ for any two columns $\mathbf{p}_{(i)}$ and $\mathbf{p}_{(j)}$; and $||\mathbf{p}_{(i)}||_{_2}=1,\ \forall \ i$.  As $\mu(\mathbf{P})$ becomes smaller, the probability of successful sparse signal recovery increases. 
Hence, we  consider the following optimization problem to design $\mathbf{P}$: 
\begin{equation}
\begin{aligned}
	&{\min}\; \; \mu(\mathbf{P})= {\min}\; \;\max_{i\neq j}|\mathbf{p}_{(i)}^{H}\mathbf{p}_{(j)}|\\
	&\text{s.t.}\; \; \;  ||\mathbf{p}_{(i)}||_{_2}=1,  
	\;\; \; \;\;1\leq i,j \leq JN(l_\tau+1).
	\end{aligned}
	\label{eq:P_opt}
	\end{equation}

To carry out the optimization in (\ref{eq:P_opt}), we consider the method of differential evolution \cite{DE_price_20}. Differential evolution is a robust evolutionary algorithm to solve arbitrary optimization
problems with real-valued parameters. The problem (\ref{eq:P_opt}) involves complex variables which are converted to the equivalent real ones of twice lengths to apply differential evolution.
Specific to the CSC model, there exists a lower bound for $\mu(\mathbf{P})$ depending on the local dictionary's dimensions \cite{csc_slice}. Considering the local dictionary $\mathbf{P}_L$ shown in Fig.~\ref{fig:ce_csc_1}, we have $\mathbf{P}_L\in \mathbb{C}^{L_p\times J}$ and the condition is obtained as 
$\mu(\mathbf{P})\geq\sqrt{\frac{J-1}{J(2L_p-1)-1}}<1$. We consider the particular case of $J=6$ and $L_p=20$ to verify the effectiveness of the proposed pilot sequence.  
The resultant dictionary optimized by differential evolution  is referred to as `Learned pilot' and it has $\mu(\mathbf{P})\approx 0.36$.  Note that for this case, we have  $0.277\leq \mu(\mathbf{P})<1$ which shows that `Learned pilot' yields a mutual coherence value close to its lower limit.  
Additionally, we analyze the following pilot sequences: (1) `Gaussian pilot' where the pilot symbols are i.i.d. complex Gaussian, (2) `Zadoff-Chu pilot' where  $\mathbf{p}_u[i]=\exp\left(-j\frac{\pi r i ^2}{L_p}\right)$ for the $u^{\text{th}}$ user, where $r$ is the order of the sequence and for each user, it is selected as a distinct prime number with $1\leq r <L_p$ and $0\leq i < L_p$, and (3) `SCMA cw pilot' where we consider SCMA codewords as the pilot vectors.
\begin{figure}[!htbp]
	\centering
	\includegraphics[height=8cm, width=17cm]{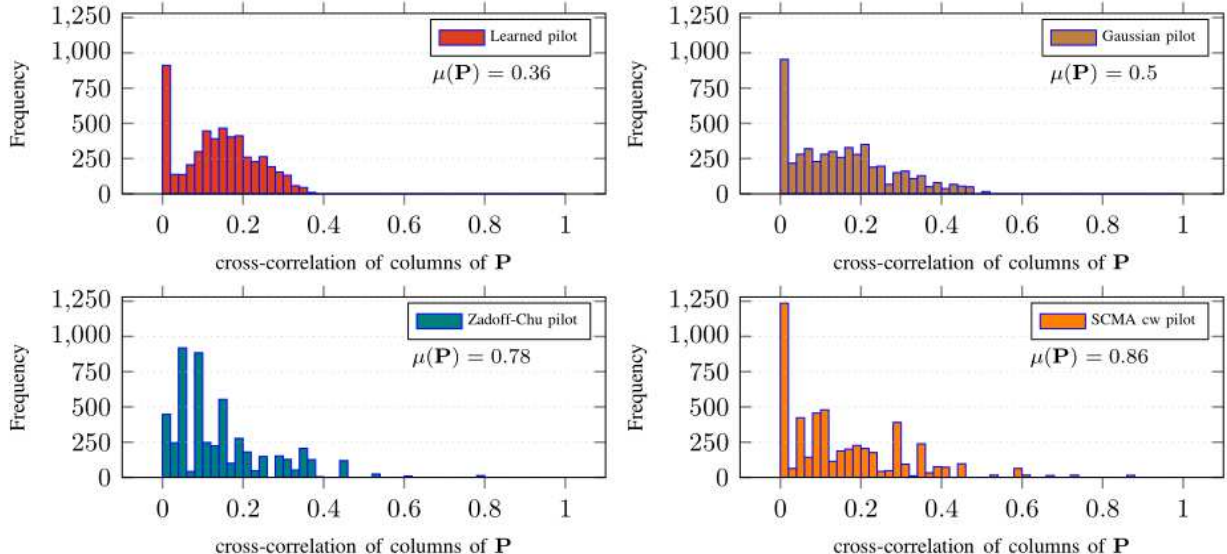}
\caption{Histogram of cross-correlation of columns of $\mathbf{P}$.}
	\label{fig:corr_P}
\end{figure}

Fig.~\ref{fig:corr_P} shows the histogram of the cross-correlation of the columns of $\mathbf{P}$ for the corresponding pilot vectors. Observe that the mutual coherence values for `Learned pilot', `Gaussian pilot', `Zadoff-Chu pilot', and `SCMA cw pilot' are 0.36, 0.5, 0.78, and 0.86 respectively. Moreover,  the spread of the cross-correlation of the columns of $\mathbf{P}$ is limited to a small range for  `Learned pilot', compared to those of  the other pilot sequences. For `SCMA cw pilot',  the spread of cross-correlation to higher values is even more. This analysis justifies the use of optimized dictionary for channel estimation. 

\subsection{Rectangular Pulse Shaping}
So far,  ideal pulse shaping  has been assumed for the proposed channel estimation technique. When rectangular pulses replace the ideal ones, the changes are reflected  only through    an additional multiplicative phase factor with the channel coefficients \cite{ravi_2018_TWC}, as shown  below:
\begin{equation}
	y_u[k,l] = \sum_{ k^\prime=0}^{N-1}\sum_{ l^\prime=0}^{M-1}{{\beta}}_{(k^\prime,l^\prime)}[k,l]{h}_{u}{[k^\prime,l^\prime]}{x}_u\left[\left[k-k^\prime\right]_N,[l-l^\prime]_M \right]
	\label{eq:otfs_basic2}
\end{equation}
where
\begin{equation}
	{{\beta}}_{(k^\prime,l^\prime)}[k,l] =\begin{cases}
		(\frac{N-1}{N})e^{j2\pi\frac{k^\prime}{N}(\frac{l-l^\prime}{M})}e^{-j2\pi(\frac{[k-k^\prime]_N}{N})}&\small{0\le l < l^\prime}\\
		e^{j2\pi\frac{k^\prime}{N}(\frac{l-l^\prime}{M})}&\small{l^\prime \le l < M}
	\end{cases}
	\label{eq:otfs_rect_ce}
\end{equation} 
Observe that neither the sparsity nor the circulant block property is affected by the rectangular pulses. Therefore, the method proposed for ideal pulse shaping can be easily extended to rectangular pulse shaping by considering the phase factor matrix  $\mathbf{\Upsilon}_{l}\in \mathbb{C}^{N\times NM} $ defined as

\begin{equation}
\mathbf{\Upsilon}_{l}[n,m]=	{{\beta}}_{(k^\prime,l^\prime)}[k,l], \ \text{for } n=k,\ m=Nl^\prime+k^\prime .\end{equation}

Now (\ref{eq:otfs_basic2})  can be written in vector form  as
\begin{align}
	\mathbf{y}_{u,{l}} &= \left(\mathbf{P}_{u}\odot\mathbf{\Upsilon}_{l}\right){\mathbf{h}}_{u,{l-\bar{l}}},\ \ \text{for } 0\le l-\bar{l}\le l_\tau, \ l=\bar{l}+l-\bar{l}
	\label{eq:io_conv_rect}
\end{align}
Note that $\mathbf{\Upsilon}_{l}$ depends only on $N$ and $M$ and is independent of a user's channel conditions. Therefore, the proposed method can be directly applied in the case of rectangular pulse, by replacing $\mathbf{P}_{u}$  by $\left(\mathbf{P}_{u} \odot \mathbf{\Upsilon}_{l}\right)$ for all users in (\ref{eq:ce_simple}).

\subsection{Complexity Analysis}\label{subsec:complexity}
The independent delay tap processing and the reduced search range in the  Doppler domain result in  remarkable dimension reduction for all matrix operations of Algorithm~\ref{algo:ce_alg}. 
\begin{table}[!htpb]
	\caption{Complexity analysis of Algorithm~\ref{algo:ce_alg}.}
	\label{table:comp_comp} 
	\small
		\begin{tabular}{|c|c|}		
			\hline
			\ &Algorithm~\ref{algo:ce_alg} \\
			\hline
			Dimension of sensing matrix&$N\times J(2k_\nu+1)$\\
			\hline
			Iterations($i$)&$\mathcal{O}(\log(J))$ \\
			\hline
			{Computational-}&\multirow{2}{*}{$\mathcal{O}(M_{\text{dom}}N(J(2k_\nu+1)+J^2)i)$}\\
			complexity &\  \\
			\hline
		\end{tabular}
\end{table}

TABLE~\ref{table:comp_comp} shows the computational complexity  based on the number of complex multiplications. In general, for the sparse signal recovery problem with a sensing matrix of dimension $U\times V$ and a sparsity of $s$, the SP algorithm's  complexity is $\mathcal{O}(U(V+s^{2}))$ per iteration.  The number of iterations needed for convergence is $\mathcal{O}(s)$ \cite{SP}.  Here we have $U=N, V=J(2k_\nu+1)$, and $s=J$. 
In TABLE~\ref{table:comp_comp}, 
$M_{\text{dom}}\ll M$ is the  number of dominant delay taps in the channel.
\begin{rem}
	Channel estimation based on the sparse signal recovery algorithms using one complete OTFS frame of full-length pilot vector has a sensing matrix of dimension $NM\times JNM$. This non-embedded scheme  results in higher complexity than Algorithm~\ref{algo:ce_alg}, in terms of the number of complex multiplications.
\end{rem}
\subsection{Cramer-Rao Lower Bound}\label{sec:crlb_limit}
The CRLB for the estimates of the channel coefficients  is presented here. We consider the  pilot symbols' arrangement as discussed in Section~\ref{subsec:data_pilot}. Writing (\ref{eq:rearrange_mat}) in matrix form, we get
\begin{equation}\left[ {\begin{array}{*{20}{c}}
	{{{\bf{y}}_{\bar l}}}\\
	{{{\bf{y}}_{\bar l + 1}}}\\
	\vdots \\
	{{{\bf{y}}_{M - 1}}}
	\end{array}} \right] = \left[ {\begin{array}{*{20}{c}}
	{{{\bf{P}}}}&{\bf{0}}& \cdots &{\bf{0}}\\
	{\bf{0}}&{{{\bf{P}}}}&{\cdots}&{\bf{0}}\\
	\vdots & \cdots & \ddots & \vdots \\
	{\bf{0}}&{\bf{0}}& \cdots &{{{\bf{P}}}}
	\end{array}} \right]\left[ {\begin{array}{*{20}{c}}
	{{{\bf{h}}_0}}\\
{{{\bf{h}}_1}}\\
	\vdots \\
	{{{\bf{h}}_{M-1-\bar l}}}
	\end{array}} \right]
\label{eq:diagonal_final}
\end{equation}
Writing (\ref{eq:diagonal_final})  compactly, we get 
\begin{equation}
{{\bf{y}}^p} ={\mathbf{P}}^{\text{diag}}\bf{h}
\label{eq:observation_vec}
\end{equation}
where ${{\bf{y}}^p}=\text{vec}\left({\mathbf{Y}}_p\right) \in {\mathbb{C}}^{N\left(M-\bar l\right) \times 1}$  contains the received symbols in the pilot observation region, ${\mathbf{P}}^{\text{diag}} =\text{diag} \left({\mathbf{P}}\right) \in \mathbb{C}^{N\left(M-\bar l\right) \times JN\left(M-\bar l\right)} $,  and ${\bf{{h}}} \in {\mathbb{C}}^{JN\left(M- \bar l\right)}$  contains the channel coefficients of all users. As the number of multi paths of each user is $P$, $\bf{{h}}$ contains $JP$ non-zero elements.  Suppose the location  of  these non-zeros are $c_1, c_2, \ldots, c_{JP}$. For deriving the CRLB, we assume that these locations are known. 
Let the $e^{\text{th}}$  non-zero element of $\bf{h}$  be denoted by $\theta_e$  ($c_e^{\text{th}}$ element of $\bf{{h}}$).
The parameter vector is given by  $\boldsymbol{\theta}=\left[\theta_1, \theta_2, \ldots, \theta_{JP}\right]^T$  which are to be estimated.   Let $p^{\text{diag}}\left[n,c\right]$ denote the element in the $n^{\text{th}}$  row and $c^{\text{th}}$  column of ${\mathbf{P}}^{\text{diag}} $ and $y^p_n$ denote  the $n^{\text{th}}$  element of ${\bf{y}}^p$. Then considering the AWGN at the BS,  (\ref{eq:observation_vec})  can be alternatively written as
$$
y_n^p =s_n+w_n,   \;\;\;\;\; \;\; n=1, 2, \ldots , N(M-\bar l)
\vspace{-0.1in}
$$
where $w_n \sim {\mathcal{CN}}\left(0,\sigma^2\right)$ and $s_n$  is given by
\begin{equation}
s_n=\sum\limits_{e = 1}^{JP} {{p^{{\rm{diag}}}}\left[ {n,c_e} \right]{\theta _e}}  \label{eq:s_l}.
\end{equation}
For  rectangular pulses, we have
\begin{equation}
{s_n} = \sum\limits_{e = 1}^{JP} {{p^{{\rm{diag}}}}\left[ {n,{c_e}} \right]\gamma \left[ {n,{c_e}} \right]{\theta _e}}   \label{eq:s_l_rect}
\end{equation}
 where $\gamma \left[ {n,{c_e}} \right]$ is given by ${{\beta}}_{(k^\prime,l^\prime)}[k,l],$ for $n=Nl+k+1$ and ${c_e}=J(Nl^{\prime}+k^{\prime})+1$.
 
As per  CRLB, the variance of any unbiased  estimator ${{\hat \theta }_i}$ is lower bounded as \cite{kay_estimation}\vspace{-0.05in}
\begin{equation}{\mathop{\rm var}} \left( {{{\hat \theta }_i}} \right) \ge {\left[ {{{\bf{I}}^{ - 1}}\left( {\boldsymbol{\theta }} \right)} \right]_{ii}}
\label{eq:crlb}
\end{equation}
where ${{\bf{I}}}\left( {\boldsymbol{\theta }} \right)$  is the $JP \times JP$  Fisher information matrix  whose element in the $i^{\text{th}}$  row and the $j^{\text{th}}$ column is given by
	\begin{equation}{\left[ {{\bf{I}}\left( {\bf{\theta }} \right)} \right]_{ij}} = \frac{1}{{{\sigma ^2}}}\sum\limits_{n = 1}^{N(M - \bar l)} {\left\{ {\frac{{\partial {s_n}}}{{\partial {\theta _i}}}\frac{{\partial s_n^*}}{{\partial {\theta _j}}} + \frac{{\partial s_n^*}}{{\partial {\theta _i}}}\frac{{\partial s_n^{}}}{{\partial {\theta _j}}}} \right\}} .
	\label{eq:fisher}
	\end{equation}	
	From (\ref{eq:s_l})  and (\ref{eq:s_l_rect}), it can be easily found that 
\begin{equation}
	\frac{{\partial {s_n}}}{{\partial {\theta _i}}}=
	\begin{cases}
		 {p^{{\rm{diag}}}}\left[ {n,{c_i}} \right]  & {\text{ for ideal pulse}}\\
		{p^{{\rm{diag}}}}\left[ {n,{c_i}} \right]\gamma \left[ {n,{c_i}} \right]   & {\text{ for rectangular pulse}}
	\end{cases}
	\label{eq:partial_dev}
\end{equation}
Similarly ${\frac{{\partial s_n^*}}{{\partial {\theta _i}}}}$ can be found out. From (\ref{eq:fisher}), the Fisher information matrix can be found out and then by using (\ref{eq:crlb}),  CRLB for the estimates of the channel  coefficients can be obtained.   Finally,  for the normalized mean square error (NMSE) analysis, we consider the average normalized CRLB:   $\left\{ {\sum\limits_{i = 1}^{JP} {{{\left[ {{{\bf{I}}^{ - 1}}\left( {\boldsymbol{\theta }} \right)} \right]}_{ii}}} } \right\}/\left\| {\boldsymbol{\theta }} \right\|_2^2$.
\section{Results and Discussions}\label{sec:simulations}
This section presents the simulation results for the proposed method of channel estimation and analyzes them. The OTFS-SCMA scheme uses an $N \times M$ DD grid of different dimensions and a basic $\left(J=6, K=4\right)$ SCMA system with 150$\%$ overloading factor.   Scheme-1 is considered where the $K \times 1$ codewords are
placed along the Doppler axis. The codebook design follows the  technique described in \cite{kuntal_codebook}. The channel conditions of all users are assumed such that only a single Doppler value is associated with a delay path \cite{ravi_2018_TWC}. Integer delay and Doppler taps are assumed in simulations. Note that the proposed method applies to the fractional cases also, albeit with an increase in complexity. For the data detection in uplink, we use the powerful single-stage MPA \cite{otfs_scma}. 
The average pilot and data SNRs are denoted by SNR$_p = \frac{\mathbb{E}(|p_{u,i}|^2)}{\sigma^{2}}$ and SNR$_d = \frac{\mathbb{E}(|{{c}}_{n,i}^u|^2)}{\sigma^{2}}$ respectively, where ${\sigma^{2}}$ is the variance of AWGN. 
The simulation observations are presented in terms  of various performance indicators like BER, spectral efficiency, NMSE,  and complexity analysis. The proposed pilot symbols  are `Learned pilot' as discussed in Section~\ref{subsec:pilot_seq}. 
We extend the  methods of OTFS-OMA in \cite{ch_est1} and \cite{ce_oma}  to OTFS-SCMA for the comparison purpose and are referred to as `QAM-pilot' and `MSP' (modified subspace pursuit), respectively.  
\subsection{BER Analysis}
For an initial investigation, we consider a DD grid with  $N=M=32$.  The number of multi-paths is $P=2$. The sub-carrier spacing is $\triangle f=1$ Hz and the symbol duration is  $T=1$s. The maximum delay tap is taken such that $l_\tau\ll M$. For each  delay tap $l_i$, the corresponding Doppler tap $k_i$ is randomly selected from $\left\{ 0,1, \ldots, k_\nu, N-k_\nu,N-k_\nu-1,  \ldots, N-1\right\}$ where $k_{\nu}\ll N$. 
 Considering $k_{\nu}\ll 32$, we take $L_p=12$ as per Lemma~\ref{rem:optimal_Lp}.

\begin{figure}[!htbp]
	\centering
		\subfloat[Ideal pulse shaping]{\includegraphics[height=6.7cm, width=8cm]{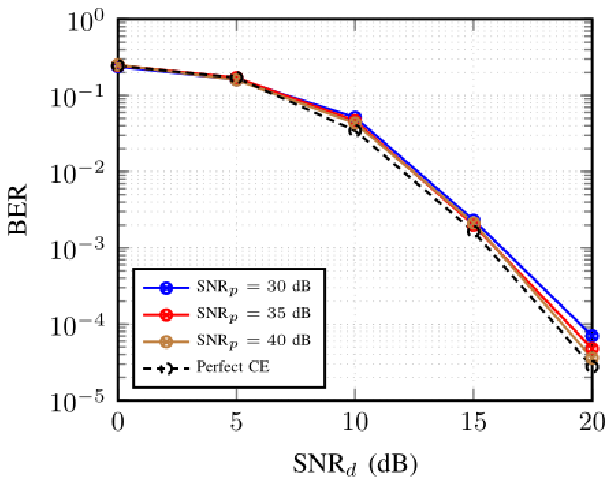} \label{fig:res_ideal}}
		\subfloat[Rectangular pulse shaping]{\includegraphics[height=6.7cm, width=8cm]{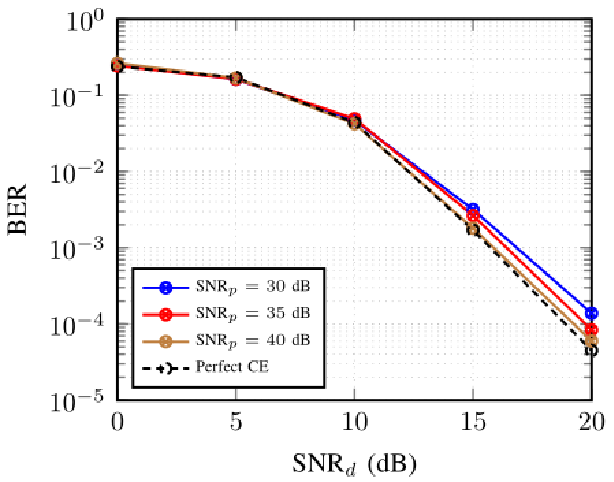}\label{fig:res_rect}		}
		\caption{\footnotesize{OTFS-SCMA uplink channel estimation using proposed method ($\Gamma_{32,32}$; $J=6,K=4$; $P=2$)}.}
			\label{fig:res_ul}
\end{figure}
 Fig.~\ref{fig:res_ul} shows the BER for three different  SNR$_p$ values of $30$, $35$, and $40$ dB  for ideal pulse and rectangular pulse. As expected, when  SNR$_p$  increases, the BER  also improves. At  SNR$_p=40$ dB, the BER nearly approaches to that of the perfect CE. 
 Observe from Fig.~\subref*{fig:res_rect} that  the presence of the additional phase factor in the case of rectangular pulse has no  noticeable impact on the BER performance  at SNR$_p=40$~dB.

   \begin{minipage}[b]{0.5\textwidth}
  	\begin{figure}[H]
  	\includegraphics[height=6.7cm, width=8cm]{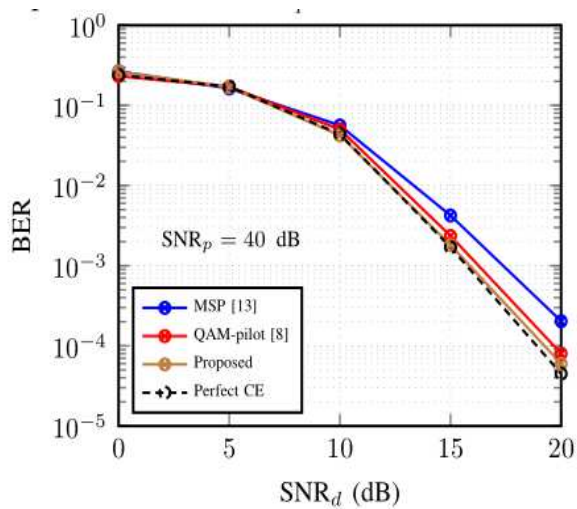}
  		\caption{\footnotesize{BER of proposed and conventional methods.}}
  		\label{fig:res_3_comp}	
  	\end{figure}
  \end{minipage}
  \begin{minipage}[b]{0.5\textwidth}
  	\begin{figure}[H]
  		\includegraphics[height=6.7cm, width=8cm]{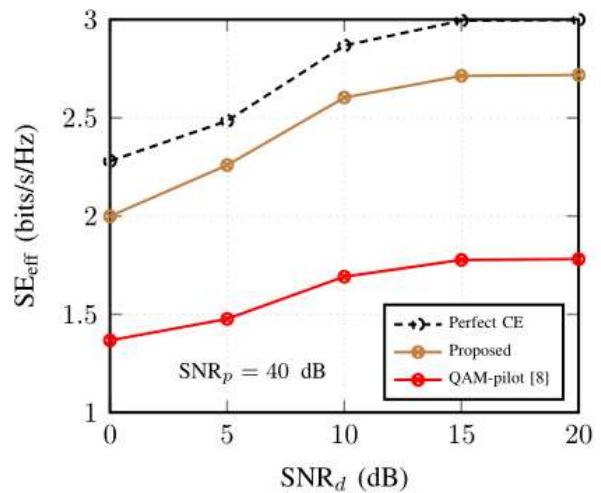}
  		\captionof{figure}{Effective spectral efficiency in uplink.}
  		\label{fig:res_5_se}
  	\end{figure}
  \end{minipage} 
 In Fig.~\ref{fig:res_3_comp}, we compare the BER performance of the proposed method to those of  `QAM-pilot' and `MSP'. The results are presented for SNR$_p=40$ dB. Observe that the proposed method gives a slightly better BER over  `QAM-pilot'. The technique of `QAM-pilot' \cite{ch_est1} uses a single QAM-pilot symbol and a thresholding-based method for detecting the paths and estimating the channel coefficients one at a time in a scalar form. On the other hand,  the proposed scheme is devised based on a sparse signal recovery algorithm with non-orthogonal pilot vectors, where the channel coefficients at a particular delay tap are estimated as a vector.  In contrast to `MSP' \cite{ce_oma}, which  is based on sparse signal recovery, we use  sparse pilot vectors which are non-orthogonal and embedded in the data frame, rather than  reserving a frame  dedicated for the orthogonal pilot vectors. The comparison in Fig.~\ref{fig:res_3_comp} shows that the proposed method of channel estimation gives an improved BER performance over `MSP'.

The sparse signal recovery of  `MSP' considers the channel coefficients of all users as a single vector. Hence, the path missed for one user will result in an additional path detected for another user to maintain the sparsity.  In effect, two users will be affected by a single path detection error, affecting the BER performance. In the proposed method, the initial estimate is taken sequentially, considering one path per user, reducing the probability of error in path detection.

 \subsection{Guard-band Overhead and Spectral Efficiency Analysis}\label{subsec:guardband}
Low guard band overhead is a desirable feature for any channel estimation technique. Since the proposed method uses an embedded pilot-aided structure, we compare its overhead  with that of `QAM-pilot' \cite{ch_est1} as displayed in TABLE~\ref{table:guardband}.  The method `MSP' in \cite{ce_oma} has a dedicated pilot frame for channel estimation alone and   the  guard band is not required. Hence it is not considered for the guard band overhead analysis. TABLE~\ref{table:guardband} shows the overheads for different user speeds  considering the EVA propagation model given in TABLE~\ref{table:eva_ch}. For each  delay tap $l_i$, the  Doppler shift  is generated using  Jakes' formula $\nu_i=\nu_{\max}\text{cos}(\theta_i)$, for $i=1,2\ldots 9$ and $\theta_i \sim  U\left(-\pi, \pi\right)$. The maximum Doppler shift  is given by $\nu_{\max}=\frac{f_c v_u}{(3\times 10^8)}$ where $v_u$ is the speed of the user. Moreover, we take $k_\nu= \left\lceil  \nu_{\max} NT \right\rceil$. The two cases of pilot vector assignment are studied: (1) reduced guard band  where the length of the pilot vectors is strictly less than $N$ and (2) full guard band where the pilot vectors occupy the entire Doppler axis. Observe from  TABLE~\ref{table:guardband} that the guard band overhead  for the proposed method is  independent of the number $J$ of users.  Moreover, the overhead for the proposed method is  less than that of `QAM-pilot' \cite{ch_est1}.  The contrast in the overheads of the two methods increases  as the   user speed increases. In the case of full guard band, the proposed method offers distinctively lower overhead than `QAM-pilot'.

The reduction in guard band overhead can also be interpreted in terms of spectral efficiency (SE). We analyze the effective SE ($\text{SE}_{\text{eff}}$) of the proposed OTFS-SCMA uplink channel estimation following   \cite{scma_ce_rev}. 
$\text{SE}_{\text{eff}}$ for a single-user is given by 
$\text{SE}_{\text{eff}}=(1-\text{BER})S 
$
where, $S$ denotes the nominal SE.

\vspace*{-.1in} 
\begin{minipage}[t][5.2cm]{0.6\linewidth}
	\hspace*{-.25in}
	\vspace*{-.3in}
	\captionof{table}{Guard band overhead for uplink.}
	\label{table:guardband}
	\vspace*{.1in}
	\scalebox{0.63}{
		\begin{tabular}{|c|c|c|c|c|c|c|c|c|}
			\hline
			{UE speed} & \multicolumn{6}{c|}{$\text{Reduced guard band}^{(a)}$} & \multicolumn{2}{c|}{Full} \\
			\cline{2-7}
			{(Kmph)}& \multicolumn{2}{c|}{30} &\multicolumn{2}{c|}{120}  & \multicolumn{2}{c|}{500}& \multicolumn{2}{c|}{$\text{guard band}^{(b)}$}\\
			\cline{1-7}
			{Max. Doppler} & \multicolumn{2}{c|}{$k_\nu = 1$}&\multicolumn{2}{c|}{$k_\nu = 4$}& \multicolumn{2}{c|}{$k_\nu = 16$}&\multicolumn{2}{c|}{ \ } \\
			\cline{2-9}
			{tap $k_\nu$} &Prop. & \cite{ch_est1} & Prop. & \cite{ch_est1} & Prop. & \cite{ch_est1}&Prop. & \cite{ch_est1} \\
			\hline
			{$\text{Pilot+Guard}$}& \multirow{2}{*}{984} &\multirow{2}{*}{1168}&\multirow{2}{*}{1476} &\multirow{2}{*}{2920} &\multirow{2}{*}{3444} &\multirow{2}{*}{9928} &\multirow{2}{*}{5248}& \multirow{2}{*}{18688}  \\
			{$\text{symbols}^{\hyperlink{Ng}{1}}$ ($N_g$)}&\ & \ & \ & \ & \ & \ & \ & \ \\
			\hline
			{$\text{Data symbols}^{\hyperlink{N_d}{2}}$}&\multirow{2}{*}{96828} &\multirow{2}{*}{96552}&\multirow{2}{*}{96090} &\multirow{2}{*}{93924}&\multirow{2}{*}{93138}  &\multirow{2}{*}{83412}&\multirow{2}{*}{90432}&\multirow{2}{*}{70272} \\
			{$\left(N_{d}\right)$}&\ & \ & \ & \ & \ & \ & \ & \ \\
			\hline
			{Guard band}& \multirow{2}{*}{1.5\%} & \multirow{2}{*}{1.7\%}& \multirow{2}{*}{2.2\%} &\multirow{2}{*}{4.5\%}& \multirow{2}{*}{5.2\%} & \multirow{2}{*}{15.2\%} &\multirow{2}{*}{8\%}&\multirow{2}{*}{28\% } \\
			{overhead} $\left(\frac{N_g}{NM}\right)$&\ & \ & \ & \ & \ & \ & \ & \ \\
			\hline
		\end{tabular}
	}
	\hspace*{-.1in}
	{\scriptsize{$^{1^{(a)}}N_g:(4k_\nu+L_{p})(2l_\tau + 1)\left[\text{Prop}\right];\quad (4k_\nu+4)(J(l_\tau+1)+l_\tau) $} `QAM-pilot' \cite{ch_est1}}
\end{minipage}
\begin{minipage}[t][5.2cm]{.4\linewidth}
	\vspace*{-.1in}
	\centering
	\captionof{table}{EVA channel model.}
	\label{table:eva_ch}
	\vspace*{-.07in}
	\scalebox{0.68}{
	\begin{tabular}{|l|c|}
		\hhline{|=|=|}
		{Parameter} & {Value}\\
		\hhline{|=|=|}
		Carrier frequency, $f_c$ & 4 \text{ GHz}\\
		Subcarrier spacing, $\triangle f$ & 15 \text{ KHz}\\
		Number of Doppler bins, $N$ &128\\
		Number of delay bins, $M$ &512\\
		Maximum delay tap, $l_\tau$ &20\\
		\hline
	\end{tabular}
	}

	\vspace*{.85cm}	
	{\scriptsize{\hspace{-3cm}$^{2}N_{d}:(NM-N_g)\frac{J}{K}$\\
			\vspace*{-.1in}
			$^{1^{(b)}}N_g:N(2l_\tau + 1)\ \left[\text{Prop}\right];\  N(J(l_\tau+1)+l_\tau)\ $}   \cite{ch_est1}}
\end{minipage}

\vspace*{-.2in}
Usually, the nominal SE  is given by $S=\frac{N_b}{N_{\text{RE}}}$,  where $N_b$ and $N_{\text{RE}}$ denote the numbers of  information bits and  resource elements (RE), respectively, in the frame. Note that  in the ideal case (BER=0), we have $\text{SE}_{\text{eff}}=S$. For an uncoded SCMA system, we have $S=\left(\frac{J}{K}\right)\log_{2}(|\mathbb{A}|)$. Extending this concept to the channel estimation of uncoded SCMA uplink system,  we get
	$S=\frac{JN_b}{KT+\lceil\frac{KT}{N_{\text{RE}}-N_g}\rceil N_g}\quad \text{ with } T=\frac{N_b}{\log_{2}(|\mathbb{A}|)}$
where $N_g$ is the number of REs reserved for channel estimation. For a typical OTFS-SCMA  system in  uplink, we have $\log_{2}(|\mathbb{A}|)=2$, $N_{\text{RE}}=NM$, $JN_b=2N_d$, and $T=\frac{N_d}{J}$. Hence the effective SE for the OTFS-SCMA uplink case becomes
\begin{equation}
	\text{SE}_{\text{eff}}=(1-\text{BER})\frac{2N_d}{NM}.
	\label{eq:se_final}
\end{equation}
Fig.~\ref{fig:res_5_se} presents  $\text{SE}_{\text{eff}}$ for the proposed method and the existing methods. 
If the  perfect CE is possible, then the system's SE  is equal to that  of an uncoded OTFS-SCMA system. Based on (\ref{eq:se_final}), for an overloading factor of 150\% and $|\mathbb{A}|=4$, we have  SE$_\text{eff}\approx 3$  bits/s/Hz. Observe from Fig.~\ref{fig:res_5_se} that the proposed channel estimation technique gives SE$_\text{eff}\approx 2.75 $ bits/s/Hz  at ${\text{SNR}}_p=40$ dB. The non-orthogonal pilot pattern improves the SE$_\text{eff}$ by almost $1  $ bits/s/Hz compared to the orthogonal `QAM-pilot' technique  \cite{ch_est1}. The simulation presented in Fig.~\ref{fig:res_5_se} is for $\Gamma_{32,32}$, $\ l_\tau=k_\nu=1$, and $L_p=12$. When we consider a practical EVA channel model, the proposed method gives an even higher SE$_\text{eff}$ advantage as presented in TABLE~\ref{table:guardband}.

\subsection{Pilot Analysis}
\subsubsection{Length Analysis}
 For the proposed method to be successful, the pilot vectors must have a minimum length $L_{p}$ satisfying the condition given in Lemma~\ref{rem:optimal_Lp}. We consider $M=512, N=128, {\text{SNR}}_d=20 \text{ dB}$, and $\text{SNR}_p=40$ dB for the analysis on the length of the pilot vectors.   
 Fig.~\ref{fig:res_9_Lp_prob}  shows the probability of success, defined as the ratio of the number of trials that resulted in successful sparse signal recovery to the total number of trials for different lengths of pilot vectors. Observe that the choice of $L_p$ influences the successful sparse signal recovery. Depending on the Doppler spread of the channel, the lower bound of $L_p$ also decreases. However, a minimum of $L_p=12$ is to be maintained irrespective of the channel conditions. The channel estimation error can be analyzed  by calculating the NMSE, which is given by $10\log_{10}\left(\frac{\mathbb{E}(||\mathbf{h}-\hat{\mathbf{h}}||^{2}_{2})}{\mathbb{E}(||\mathbf{h}||^{2}_{2})}\right)$  dB. We plot the NMSE values   against the pilot lengths in Fig.~\ref{fig:res_9_Lp}. For the successful channel estimation, even at a maximum user velocity of $500$ Kmph, we can select $L_p=20$. Further, from Fig.~\ref{fig:res_9_Lp}, it is evident that to achieve a tolerable NMSE, a minimum pilot vector length of $20$ is  necessary. For $L_p \geq 20$, the NMSE remains almost constant while the guard band overhead increases by around $2.5\%$. Also, from the CRLB curve, it is clear that  for $L_p \geq 20$, the proposed method performs close to the lower bound specified by CRLB. The simulation results shown in Fig.~\ref{fig:res_9_Lp_prob} and Fig.~\ref{fig:res_9_Lp} substantiate Lemma~\ref{rem:optimal_Lp} and its proof given in Appendix~\ref{app:optimal_lp}.

\begin{minipage}[b]{0.5\textwidth}
	\begin{figure}[H]
			\includegraphics[height=6.7cm, width=7cm]{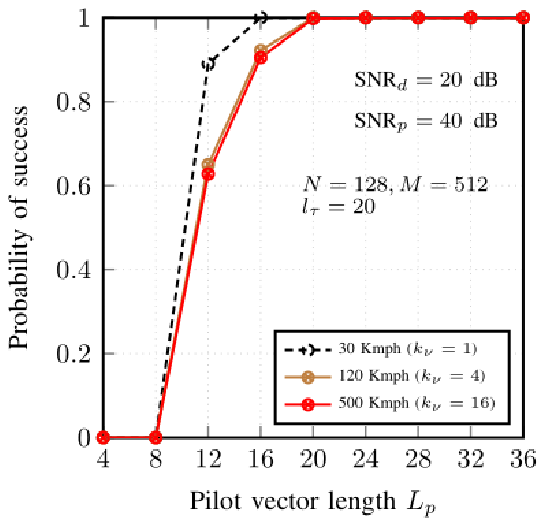}
			\vspace{-0.35in}
	\end{figure}	
	\captionof{figure}{Probability of success for different user velocities.}
	\label{fig:res_9_Lp_prob}
	
\end{minipage}
\begin{minipage}[b]{0.5\textwidth}
	\begin{figure}[H]
		\includegraphics[height=6.7cm, width=8cm]{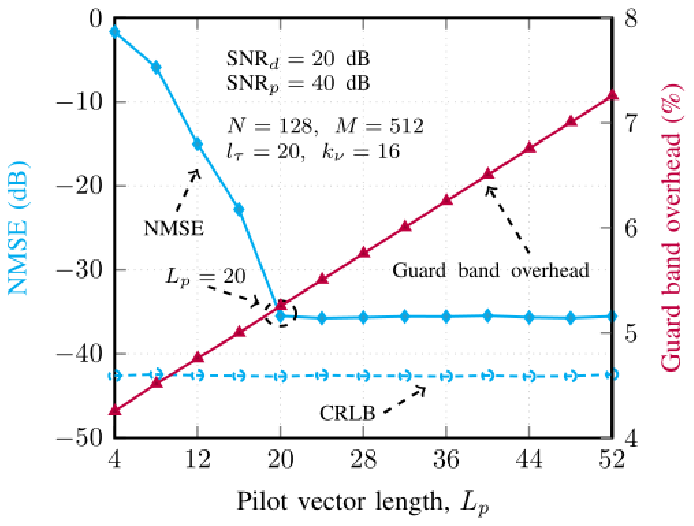}		
	\end{figure}
	\vspace*{-1cm}
	\captionof{figure}{Optimal pilot vector length of proposed method.}
	\label{fig:res_9_Lp}
\end{minipage} 
   
\subsubsection{Pattern Analysis}
 In Fig.~\ref{fig:res_6_nmse_all}, we analyze  the NMSE  for the pilot-data arrangement shown in Fig.~\ref{fig:tx_rx_grid}, using different types of pilot vectors as discussed in Section~\ref{subsec:pilot_seq}. Considering $k_{\nu}=16\ll 128$, we take $L_p=20$ as per Lemma~\ref{rem:optimal_Lp}. The differences in NMSE for these various pilot vectors follow the  mutual coherence analysis presented in Section~\ref{subsec:pilot_seq}. The proposed 
 `Learned pilot'  provides a lower NMSE than `Gaussian-pilot',  as observed in Fig.~\ref{fig:res_6_nmse_all}.  
 `Zadoff-Chu pilot'  gives a poor performance compared to `Learned pilot' and `Gaussian-pilot'. The  NMSE  values for `SCMA cw pilot'  are significantly higher than those of other sequences. The  `Random-SP init'  curve shows NMSE results when the initialization in (line~8 of Algorithm~\ref{algo:ce_alg}) is done randomly. By taking the initial estimates sequentially, we can reduce the estimation error. Also, the NMSE performance using  non-sparse i.i.d Gaussian random vectors is shown by the `Non-sparse Gaussian' curve. The slight performance improvement of the sparse pilot vector over the non-sparse one is  attributed to minimal multi-user interference in the pilot observation region.

\begin{minipage}[b]{0.5\textwidth}
	\begin{figure}[H]
		\includegraphics[height=6.7cm, width=8cm]{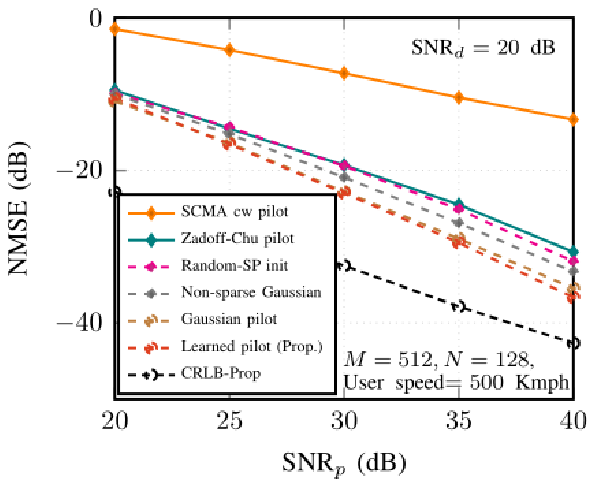}
	\end{figure}
		\captionof{figure}{NMSE for various pilot vectors.}
	\label{fig:res_6_nmse_all}
\end{minipage}
\begin{minipage}[b]{0.5\textwidth}
	\begin{figure}[H]		
		\includegraphics[height=6.7cm, width=8cm]{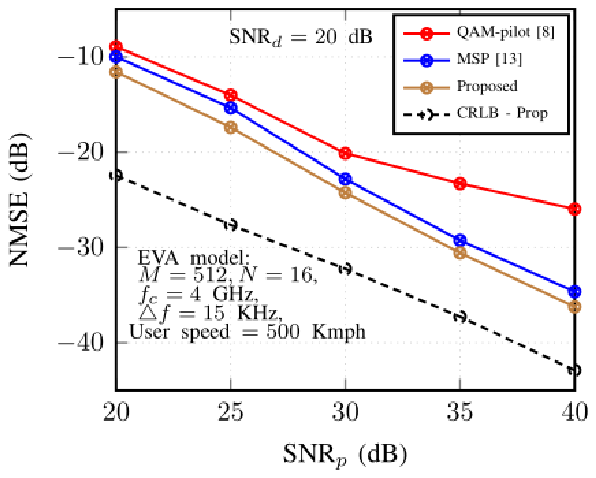}
		\end{figure}
	\captionof{figure}{NMSE of channel estimation in uplink.}
	\label{fig:res_6_nmse}
\end{minipage}

\subsection{NMSE Comparison with Existing Methods}
In Fig.~\ref{fig:res_6_nmse}, we compare the NMSE of different channel estimation techniques. Considering $k_{\nu}\ll 16$, we take $L_p=12$ as per Lemma~\ref{rem:optimal_Lp}. Observe the significant difference between the NMSE values for the sparse signal recovery-based techniques and  `QAM-pilot'. The proposed method  and `MSP' are better than `QAM-pilot' in terms of NMSE since they are based on estimation of $\mathbf{h}$ as a vector.
Moreover, the proposed method gives slightly better NMSE  than `MSP'. This improvement may be attributed to the CSC model acting independently  to each delay tap and  the sequential initial estimate method giving equal preference to all users.  For  `QAM-pilot' of~\cite{ch_est1}, an error floor occurs at higher SNR$_p$. The error floor issue arises as  this method detects a path by thresholding the components scalar-wise and then evaluating the corresponding channel coefficient rather than estimating $\bf{h}$ as a vector.  The CRLB derived in Section~\ref{sec:crlb_limit} is also shown in Fig.~\ref{fig:res_6_nmse}.


\subsection{Complexity Analysis}\label{subsec:sim_complexity}
Fig.~\ref{fig:res_complexity} highlights the complexity  of the proposed CSC modeling in terms of the number of complex multiplications. The legend `OMP' refers to the direct OMP method discussed in \cite{ce_oma}. Observe from TABLE~\ref{table:comp_comp} and Fig.~\ref{fig:res_complexity} that the proposed method is superior in two aspects of complexity:~(i)~ For a given DD grid of dimension $N$ and $M$, the proposed method's complexity depends only on the maximum Doppler $k_\nu$ of the channel. It is observed that up to  $k_\nu+1 \approx \frac{N}{2}$, the proposed method offers a noticeable complexity reduction; ~(ii)~ The complexity of proposed method is completely independent of the delay dimension $M$, since  each delay path is processed independently. On the other hand,  the conventional sparse signal recovery methods always have a constant maximum complexity independent of the Doppler conditions of the channel.
Fig.~\ref{fig:avg_iterations} presents the analysis of the average number of iterations required for the convergence of Algorithm~\ref{algo:ce_alg} for different values of the number $P$ of multipaths. Observe  that the number of iterations required is independent of $P$. The legends `Prop-ideal' and `Prop-sim'  show the  average number of iterations required for the proposed method as per TABLE~\ref{table:comp_comp} and the simulations respectively. These two curves exhibit high correlation and remain at the same level as $P$  increases.  On the other hand, for `OMP' and `MSP' algorithms, the number of iterations increases with $P$.

\begin{minipage}[b]{0.5\textwidth}
	\begin{figure}[H] 	
	\includegraphics[height=6.7cm, width=8cm]{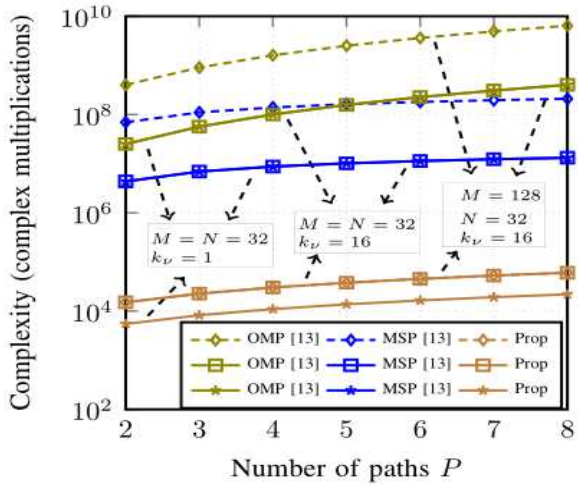}
	\captionof{figure}{Complexity analysis of proposed method.}
	\label{fig:res_complexity}
\end{figure}
\end{minipage}
\begin{minipage}[b]{0.5\textwidth}
	\begin{figure}[H] 
			\includegraphics[height=6.7cm, width=8cm]{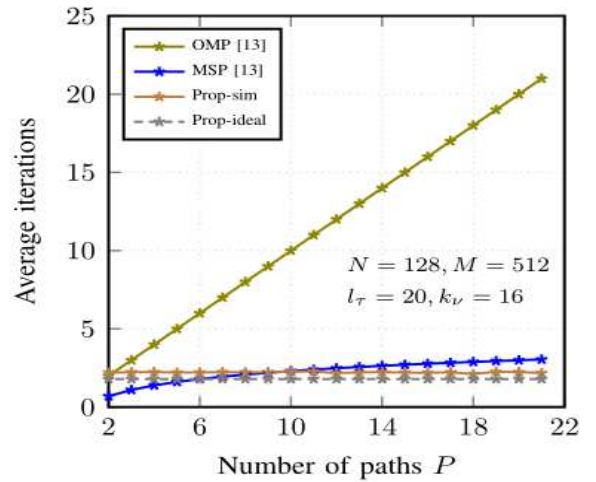}
	\captionof{figure}{Average iterations for convergence of Algorithm~\ref{algo:ce_alg}.}
	\label{fig:avg_iterations}	
	\end{figure}
\end{minipage}
\section{Conclusion}  \label{sec:conc}
In this paper, we presented a channel estimation technique for OTFS-SCMA based on CSC. The proposed method considers sparse vectors of optimum length as pilots, which follow the same sparsity pattern of the SCMA data codewords of all users. The sparse pilot vectors are designed such that the mutual coherence of the corresponding dictionary is minimized. The non-orthogonal arrangement
of the pilot vectors ensures that the guard band overhead is minimal and does not escalate with the number of users. Moreover, the proposed CSC model converts the overall problem to the channel estimation at each delay tap. This reduction of
the dimensionality, in turn, lowers the complexity of the sparse recovery algorithm significantly. Finally,  impressive BER performances and spectral efficiency are obtained, maintaining a reduced guard band overhead and complexity.  These results corroborate the proposed method's suitability for a high Doppler uplink scenario.

\appendices
\section{Proof of Lemma 1}\label{app:optimal_lp}

Consider the channel estimation problem defined in (\ref{eq:rearrange_mat}):
\begin{equation}
	\mathbf{y}_{l}=\mathbf{P}{\bf{{h}}}_{l-\bar{l}}\ \ \;\;\;\text{for } l=\bar{l}, \ \bar{l}+1,\ldots , M-1
	\label{eq:rearrange_mat_ap}
\end{equation}
where ${\bf{{h}}}_{l-\bar{l}}$ is the sparse signal to be recovered. The minimum value of $L_{p}$ is derived considering the following points:
\begin{enumerate}
	\item For Algorithm~\ref{algo:ce_alg}, the maximum possible number of columns in $\mathbf{P}_{{(\tilde{S}^{(i)})}}$ is $2J$. For any matrix $\mathbf{A}$, we have $\mathbf{A}^{\dagger} = (\mathbf{A}^{H}\mathbf{A})^{-1}\mathbf{A}^{H}$ and $\mathbf{A}^{\dagger}\mathbf{y}$ results in a solution having the number of non-zero elements as $ \text{rank}(\mathbf{A}^{\dagger})$. Since $\mathbf{P}_{{(\tilde{S}^{(i)})}}^{\dagger}\mathbf{y}_l$ (line~13 of Algorithm~\ref{algo:ce_alg}) has to recover $2J$ elements in each iteration the following condition must be satisfied: 	
\begin{equation}
	\left.\begin{aligned}	
	\text{rank}({\mathbf{P}^{\dagger}_{{(\tilde{S}^{(i)})}}})\geq 2J \ \Rightarrow L_p\geq 2J .
\end{aligned}\right.
\label{eq:cs_rank}
\end{equation}
\item Successful sparse signal recovery is guaranteed  if a minimum number of observations are available, as discussed in \cite{cs_eldar} and \cite{cs_review}. In general, for a signal having sparsity of $\frac{u}{U}$, for successful recovery, the number $V$ of observations must satisfy $V \ge \lceil cu\log U\rceil$ with $1<c\le2$. Specific to (\ref{eq:rearrange_mat_ap}), $V$    is limited by the Doppler spread of the channel.  Assuming the presence of maximum Doppler paths in the channel, we have $V=L_p+2k_\nu$ and for successful recovery, the following condition must be satisfied:
\begin{equation}
	L_p  \ge \lceil cJ\log(J(2k_\nu + 1))\rceil-2k_\nu.
\label{eq:cs_obs_1}
\end{equation}
For the EVA parameters of TABLE~\ref{table:eva_ch} ($k_\nu=16$)  with the minimum $L_p=12$ given by (\ref{eq:cs_rank}), we have from  (\ref{eq:cs_obs_1}):
		$c=\frac{L_p+2k_\nu}{J\log(J(2k_\nu + 1))}\approx 1.4$.
Thus $L_p=12$ satisfies the condition in (\ref{eq:cs_obs_1}).  However, if  the maximum Doppler paths are absent in the channel, then a higher value of $L_p$  is desirable as observed in the simulation results shown in Fig.~\ref{fig:res_9_Lp}.
%
\item For the CSC problem,  the global sparsity $s$ and the  local sparsity $s_L$ are related as $s\approx \frac{s_L}{2p-1}U$, where $p$ is the length of a dictionary element and $U$ is the length of the sparse signal to be recovered \cite{csc_slice}. Specific to (\ref{eq:rearrange_mat_ap}), we have a fixed sparsity $s=J$, the length  $L_p$ of the pilot vector corresponds to $p$, and $U$ is $J(2k_\nu+1)$. Considering the  uniform distribution of the Doppler taps, we take $s_L=1$. The sparsity relation can be expressed as
\begin{equation}
	\left.\begin{aligned}
		&J \approx\frac{1}{2L_p-1}J(2k_\nu+1) \Rightarrow
      L_p \approx k_\nu+1 .\\
	\end{aligned}\right.
	\label{eq:csc_lp_proof}
\end{equation}
For EVA parameters, we get  $L_p \approx 17$. 
\item Since the pilot vectors  follow the same  sparsity structure of the codewords with length $K$, we must have  $[L_p]_K=0$.	 
\end{enumerate}
To minimize the guard band overhead, $L_{p}$ is chosen as the nearest integer satisfying the above conditions:
\begin{equation}
L_p > \max\left\{2J,  \left\lceil cJ\log(J(2k_\nu + 1))\right\rceil -2k_\nu,   k_\nu+1 \right\}  \;\;\;  \text{ with } \left[L_p\right]_K=0\text{ and }1<c\le2 .
\label{eq:Lp_proof}
\end{equation}
\bibliographystyle{ieeetran}
\footnotesize
\bibliography{references_OTFS_SCMA}
\end{document}